\documentclass[12pt, onecolumn, journal]{IEEEtran}
\usepackage{setspace}
\IEEEoverridecommandlockouts
\usepackage{cite}
\usepackage{amsmath,amssymb,amsfonts}
\usepackage{algorithmic}
\usepackage{graphicx}
\usepackage{textcomp}
\usepackage{xcolor}
\usepackage[small]{caption}
\usepackage{subcaption}
\usepackage{amsthm}
\usepackage{breqn} 
\usepackage{setspace}
\usepackage{hyperref}
\usepackage{booktabs}
\usepackage{siunitx}
\usepackage{tikz}
\usetikzlibrary{automata, positioning, arrows}

\newtheorem{proposition}{Proposition}

\newtheorem{remark}{Remark}

\newcommand{\Bras}[1]{\bigg[ #1 \bigg]}
\newcommand{\Brap}[1]{\bigg( #1 \bigg)}
\newcommand{\brac}[1]{\left\{ #1 \right\}}
\newcommand{\bras}[1]{\left[ #1 \right]}
\newcommand{\brap}[1]{\left( #1 \right)}

\newcommand{\Exp}{\mathbb{E}}

\newcommand{\sZ}{\mathbb{Z}_{+}}
\newcommand{\sR}{\mathbb{R}_{+}}
\newcommand{\sZZ}{\mathbb{Z}_{++}}
\renewcommand{\Pr}[1]{\mathbb{P}\brac{#1}}

\newcommand{\LBterm}[1]{\frac{\delta(\tau_{max}+\Exp{\tilde{G}})}{P(#1)#1 - P(\tau_{max})\tau_{max}}}
\newcommand{\sysA}{\mathbb{A}}
\newcommand{\sysB}{\mathbb{B}}

\def\BibTeX{{\rm B\kern-.05em{\sc i\kern-.025em b}\kern-.08em
    T\kern-.1667em\lower.7ex\hbox{E}\kern-.125emX}}

\begin{document}

\title{Tradeoff of age-of-information and power under reliability constraint for short-packet communication with block-length adaptation}

\author{Sudarsanan A. K., Vineeth B. S., and Chandra R. Murthy
\thanks{Sudarsanan A. K. and Vineeth B. S. are with the Department of Avionics, Indian Institute of Space Science and Technology, Trivandrum and Chandra R. Murthy is with Department of Electrical Communication Engineering, Indian Institute of Science, Bangalore. Emails: {sudarsanan.ak.1999@gmail.com}, {vineethbs@gmail.com} and {cmurthy1@gmail.com}.}
}
\maketitle
\begin{abstract}
In applications such as remote estimation and monitoring, update packets are transmitted by power-constrained devices using short-packet codes over wireless networks.
Therefore, networks need to be end-to-end optimized using information freshness metrics such as age of information under transmit power and reliability constraints to ensure support for such applications.
For short-packet coding, modelling and understanding the effect of block codeword length on transmit power and other performance metrics is important.
To understand the above optimization for short-packet coding, we consider the optimal tradeoff problem between age of information and transmit power under reliability constraints for short packet point-to-point communication model with an exogenous packet generation process.
In contrast to prior work, we consider scheduling policies that can possibly adapt the block-length or transmission time of short packet codes in order to achieve the optimal tradeoff.
We characterize the tradeoff using a semi-Markov decision process formulation.
We also obtain analytical upper bounds as well as numerical, analytical, and asymptotic lower bounds on the optimal tradeoff.
We show that in certain regimes, such as high reliability and high packet generation rate, non-adaptive scheduling policies (fixed transmission time policies) are close-to-optimal.
Furthermore, in a high-power or in a low-power regime, non-adaptive as well as state-independent randomized scheduling policies are order-optimal.
These results are corroborated by numerical and simulation experiments.
The tradeoff is then characterized for a wireless point-to-point channel with block fading as well as for other packet generation models (including an age-dependent packet generation model).
\end{abstract}
\section{Introduction}
Remote estimation and monitoring of relevant system processes are becoming increasingly important in smart cities, internet-of-things (IoT), and industrial IoT for various applications such as environmental monitoring, feedback control and actuation, and security\cite{8930830}. 
Wireless networks for such applications have to be end-to-end optimized for information freshness \cite{6195689} (for 
instance using age of information) rather than for conventional metrics such as delay or throughput.
The majority of traffic generated in such networks for freshness-sensitive applications comprises short packets  
\cite{durisi2016toward}, \cite{xiang2020noma}. 
In order to transmit these short packets over noisy channels, short packet block codes (SPC) with \emph{smaller} codeword 
lengths are employed whose finite block length reliability is a major concern.
Transmission power is also another concern in these battery-constrained monitoring systems.
Understanding the tradeoff and interplay between such key performance indicators (KPIs) such as age of information (AoI), transmit power, and reliability is important for designing modern energy-efficient and reliable next generation wireless networks \cite{cao_metaverse}.

In order to understand the tradeoff between the above KPIs, we consider a point-to-point link in this paper.
A natural question that arises in this context is how the tradeoff between AoI, transmit power, and reliability can be achieved for a point-to-point link.
The tradeoff can be achieved by dynamic scheduling, including dynamic scheduling of packet generation, transmission times, durations, and/or rates, as well as transmit power.
A large body of work has studied such scheduling policies and associated tradeoffs but under the assumption that codeword block lengths are large \cite{neely}.
However, for SPC, modelling the relationship between reliability, transmit power, and codeword length is important for understanding the tradeoffs achieved by scheduling.
The use of codeword block-length adaptation in achieving such tradeoffs has not received much attention in prior work and is the focus of our paper.
We consider dynamic scheduling policies that adaptively chose the codeword length $\tau$ with a corresponding transmit power $P(\tau)$ so that every transmission satisfies a reliability constraint.
A major contribution is that we identify scenarios where scheduling policies that are non-adaptive, i.e., which use a fixed block codeword length achieve a close-to-optimal tradeoff between average AoI and transmit power for a fixed reliability.
We also contribute novel analytical upper and lower bounds on the achievable tradeoff between average AoI and transmit power.
We now review relevant prior work in this area.

\noindent\textbf{Prior work:}
The tradeoff of average transmit power and AoI has been considered in a number of papers. 
In \cite{9483628}, the authors consider tradeoff of AoI and the total energy consumption as a constrained Markov decision process (CMDP) and solve it using Lagrangian relaxation. The tradeoff between the AoI, quality/distortion, and energy is considered in \cite{9377455}.
Yifan et al. \cite{gu2019timely} consider an Internet-of-Things (IoT) scenario where nodes send status updates through 
an unreliable fading channel using a truncated ARQ scheme.
Closed form expressions for the AoI are derived and transmit power optimization is done.
An online greedy algorithm is developed to minimize a linear combination of quality metric, AoI, and the energy cost.
The tradeoff between age and quality/distortion is analyzed in terms of age-dependent distortion constraints in \cite{9524846}. 
Energy minimization under a peak AoI constraint is considered in \cite{saurav2021online}, where the packets can be selected/ deselected for service and the transmission rate can be chosen based on the current AoI to satisfy the AoI constraint.
In \cite{2204.02953}, an optimal non-preemptive policy that minimizes a linear combination of weighted AoI and total service cost in a G/G/1 queuing system with a single server by transmitting potentially a subset of updates is developed.
The energy-age tradeoff in a status update system with feedback having packet losses is considered in \cite{1808.01720}. A threshold-based retransmission policy with a constraint on the maximum allowed retransmissions of a packet is analyzed, and closed-form expressions for the average AoI and energy consumption are derived.
A two-threshold  (one on the AoI at the transmitter and the other on the AoI at the receiver) optimal stationary policy for the energy-age tradeoff for the status updates from a sensor to a monitor over an error-prone channel with feedback on transmission success/failure is proposed in \cite{two-threshold}. The sensor can choose to sleep, sense and transmit a new update or re-transmit, considering sensing and transmission energy. A discounted cost problem is formulated with the cost being a linear combination of average AoI and average energy consumption.  
In \cite{HARQ}, the transmitter can either re-transmit the existing data to save energy or sense and transmit new data to 
reduce AoI if a transmission fails due to channel impairments. Hybrid Automatic Repeat request (HARQ) is used for the 
feedback mechanism, and a threshold-based retransmission policy is adopted.
We note that the tradeoff between age and energy have also been considered in contexts such as energy harvesting 
\cite{bacinoglu2018achieving}, and with sensing energy \cite{gong2018energy}.
In contrast to this body of work, our paper explicitly models the finite block-length of transmissions and its effect on reliability.

Recently a number of papers have considered the reliability of finite block length codes in the context of AoI. 
Such a consideration is especially important in status update systems where the amount of data encoded and 
sent is only of the order of hundreds of bytes.
The average AoI-energy tradeoff for a short packet-based status update system with retransmissions in a broadcast scheme 
over an error-prone channel is considered in \cite{xie2021age}. 
Minimisation of the ratio of average AoI and energy using an optimal static choice of packet length is introduced. 
In \cite{xie2021_twohop}, the age-energy tradeoff for two-hop decode-and-forward relaying networks based on short packets is investigated and the tradeoff is achieved by minimizing the weighted sum of the average AoI and the average energy cost.
Yu et al. \cite{yu2020average} consider the SPC nature of communication. 
The block length used for coding is fixed during the operation of the system. 
Closed form expressions for AoI are obtained for ARQ schemes as well as schemes which discard packets. 
The optimal fixed block length is obtained. Yu et al. \cite{yu2020joint} also considers
an extension to a centralized multiple access scenario for status update systems.
The reliability of transmissions using finite block length codes was also considered for a cognitive status update 
system in \cite{zhu2022short}.
Wang et al. \cite{wang2019age} consider finite block length coding with different packet management strategies at the 
transmitter, namely pre-emptive and non pre-emptive schemes.
They obtain the average AoI for different schemes as a function of the block length; the block length is also optimized 
for minimizing the AoI.
Tang et al. \cite{tang2022average} consider an extension of the above study to non-linear age functions.
Reliability due to finite block length coding is also considered.
Preemptive strategies are found to have worse non-linear age performance.
The impact of finite block length on delay and age violation probability was investigated in \cite{devassy2018delay}. 
The authors obtained a block length that minimizes the delay and age violation probabilities.
The impact of finite block length coding for age and freshness related metrics have also been investigated in other 
scenarios such as UAVs \cite{zhang2021minimizing}, closed-loop control \cite{cao2023age}, multiple access 
\cite{sung2021age}.
The tradeoff between delay and age in a finite block length regime was considered in \cite{cao2021can}.
A static optimization of the block length as well as packet update rate for optimally trading off age with delay is 
carried out in the paper.
We note that these papers the block length is kept fixed during the operation of the system and is considered to be a 
parameter that can be statically optimized for performance enhancement.
In contrast, our work considers the dynamic adaptation of the codeword length especially to tradeoff the transmit power 
with age with a fixed reliability.

The above dynamic adaptation is an important degree-of-freedom for trading off AoI, transmit power, and reliability.
In the transmission of short packet codewords, for a given reliability of codeword transmissions, the per-packet 
transmission duration and transmit power can be traded off with one another \cite{5452208}.
However, only a few papers have considered such an approach.
Early work in dynamic adaptation of codeword length is found in \cite{vineeth_wiopt} and \cite{uysal}.
The authors in \cite{vineeth_wiopt} considered the design of scheduling policies for minimizing average delay given a 
per-codeword error constraint.
The policies adapted the rate of transmission; for a given codeword error probability, the dependence of rate of 
transmission on codeword length was modelled using channel error exponents.
Uysal et al. \cite{uysal} considered the dynamic control of packet transmission durations in order to tradeoff delay and 
transmit power, with longer transmission durations requiring lower power (and vice versa).
Recently, Zhao et al. \cite{zhao_queueurllc} considered adaptation of codeword allocation in both time and frequency 
resources for minimizing latency in an URLLC scenario.  
In these papers, the dynamic adaptation of codeword length or transmission duration has been considered for the case of 
delay.
In contrast, we consider age of information which is a more appropriate metric for applications such as 
remote estimation. 
Motivated by these, in this paper we investigate transmitter control policies which dynamically choose the duration 
$\tau$ over which each packet is transmitted in order to adapt its per-packet transmit power 
$P(\tau)$.\footnote{The power $P(\tau)$ can be modelled as a convex non-increasing function of $\tau$, see Section 
\ref{section:tau_vs_p}.}

Adaptation of block codeword length in the context of age of information has been considered in \cite{han2019optimal}, 
\cite{baoquan_adapt}, and \cite{liu2019taming}.
Han et al. \cite{han2019optimal} consider a multiple access scenario where multiple users use a TDMA frame to 
communicate.
The time durations allotted to the different users could be different and can be adapted dynamically from frame to 
frame.
The time durations correspond to codeword lengths and affect the reliability of transmissions.
The problem of allocating time durations is posed as a MDP and solved.
Liu and Bennis \cite{liu2019taming} consider a problem of minimizing average power subject to constraints on the 
average of age cost function evaluated at transmission instants and maximal average age over all time (which is related 
to peak age of information).
This is a age-power tradeoff problem which is similar to what is considered here.
The authors assume that every transmission has a fixed reliability of $\epsilon$ and the power, block codeword length, 
as well as time to the next sample can be controlled dynamically at every transmission instant.
The power and block codeword length is chosen in order to satisfy the reliability constraint as well as support a fixed 
packet size.
We note that a critical assumption made in this work is that the transmission duration, which depends on the block 
codeword length is negligible compared with the intersampling durations and that the age resets to zero after every 
successful transmission.
This allows the authors to limit the effect of the block codeword length on the age evolution via reliability.
Such an assumption is valid in the regime where the intersampling durations are large compared with the maximum 
codeword lengths used.
In contrast, we allow intersampling durations and codeword lengths to be comparable in our analysis.
We also note that we do not control the intersampling durations and our metric is the average age rather than a 
constraint on maximal age or an age function.
We note that Liu and Bennis's policy which is based on stochastic Lyapunov optimization yields a policy with a fixed 
power and codeword allocation if the intersampling duration is fixed and not optimized - which is the same as our 
result - and shows that Lyapunov optimization is a good approach.   
We note that adaptation of block codeword length for the tradeoff of age and power has been considered by 
Yu et al. \cite{baoquan_adapt}.
The authors consider a point to point system where packets are generated periodically with the period being the 
coherence time of a block fading channel. 
Packet transmissions are done at a constant power without any constraint placed on reliability as the authors consider 
applications to non-critical scenarios.
The authors use a constrained Markov decision process framework and obtain stationary randomized policies for 
minimizing average age subject to an average power constraint.
We note that our work considers a similar tradeoff problem however under a per-transmission reliability constraint.
We also consider different packet generation models.

\noindent\textbf{Contributions:}
We consider a point-to-point link model to understand the tradeoff between AoI, transmit power, and reliability.
We start with a simplified packet generation model, where we assume that packets are generated according to an independent process with a packet generation rate of $\lambda$.
As stated before, an important feature of our model is that the transmission time $\tau$ (which correspond to the block codeword length) can be dynamically adapted with a corresponding choice of transmit power $P(\tau)$ such that a reliability constraint is met for each transmission.
We formulate the optimal tradeoff problem between AoI and average transmit power as a semi-Markov decision problem where at each transmission time, a scheduler decides on a transmission time $\tau$.

The major contribution of this work is the identification of certain regimes of operation in which scheduling policies that are either non-adaptive (i.e., which always choose a fixed value of $\tau$ for all transmissions) or which are \emph{state-independent randomized policies} achieve tradeoffs that are close to optimal.
This regime consists of point-to-point links with high reliability and high packet generation rate.
Thus, we show that the use of non-adaptive policies in \cite{xie2021age}-\cite{cao2021can} are relevant.

Other contributions include:
\begin{enumerate}
\item We obtain analytical upper bounds and lower bounds on the average AoI-power tradeoff. These lower bounds are useful in analyzing the performance of scheduling policies. 
\item We show that show that the class of non-adaptive or state-independent randomized policies are \emph{order-optimal} (a weak form of optimality) in a low-power and high-power regime. The above analytical bounds are also used for this.
\item In a system with low packet generation rate, adaptive policies such as a \emph{threshold policy} achieves better tradeoff performance. We obtain an analytical approximation for the tradeoff for threshold policies for high-reliability systems.
\end{enumerate}

For non-adaptive or fixed transmission time policies, we also consider extensions to the following cases:
\begin{enumerate}
\item We consider a point-to-point link with block fading and characterize the tradeoff and the effect of channel coherence time on the tradeoff.
\item We extend to other packet generation models. We consider a model where packets are generated on the basis of an age threshold and also a \emph{preemptive} packet generation model. Analytical characterizations of the tradeoff are obtained for these cases.
\end{enumerate}
We note that the tradeoff problem for a pre-emptive packet generation model without errors was considered in our prior work \cite{sudarsanan2023optimal}.
In contrast, this paper considers a case with transmission errors as well as other packet generation models.

\noindent\textbf{Notation:}
We use the following notation: (i) $f(x)$ is $\mathcal{O}(g(x))$ if there exists a $c > 0$ 
such that $\lim_{x \rightarrow 0} \frac{f(x)}{g(x)} \leq c$; $f(x), g(x) \geq 0$, and (ii) $f(x)$ is $\Omega(g(x))$ if 
there exists a $c > 0$ such that $\lim_{x \rightarrow 0} \frac{f(x)}{g(x)} \geq c$; $f(x), g(x) \geq 0$.
Sequences which are monotonically increasing to a limit point are denoted using $\uparrow$, while those monotonically 
decreasing are denoted using $\downarrow$.
We denote the set of non-negative integers and non-negative real numbers by $\sZ$ and $\sR$ respectively.

\section{System Model and Problem Statement}
\label{section:System_Model}

We consider a time-slotted model with slots indexed by $t \in \sZ$.
New status update packets are generated at the transmitter according to the following random process.
A packet is generated with probability $\lambda$ at the start of the first slot or at the start of any slot after a packet transmission ends\footnote{Here we assume that $\lambda$ is a given quantity. We study the case where $\lambda$ can be optimized in Section \ref{section:Pre}.}.
This generation is independent of any other event.
If the packet is not generated, then the process repeats in the next slot with probability $\lambda$ until a packet is 
generated.
Thus, there is a random Geometric$(\lambda)$ delay between the end of transmission of a packet and the generation of a 
new packet\footnote{We discuss another packet generation model in Section \ref{section:Pre} in which the packet 
generation process is assumed to be an independent and identically distributed Bernoulli process $(U[t], t \in \sZ)$, 
with $U[t] = 1$ indicating that a new packet is generated at the beginning of slot $t$. When a new packet is 
generated, there are two options - either it can be discarded, in which case the model is the same as that which is 
considered here, or the new packet can pre-empt an ongoing transmission. The latter is discussed in Section 
\ref{section:Pre}.}. 
We denote this Geometric$(\lambda)$-distributed generation delay that takes values in $\{0, 1, 2, ...\}$ by $\tilde{G}$. 
We note that new packets are not generated whenever a packet transmission is ongoing.
We index packets using $m \in \sZ$.
Packets are assumed to be of fixed length of $K$ bits.
The slot in which the $m$\textsuperscript{th} packet is generated at the transmitting node is denoted by $T[m]$; $T[0] \sim$ Geometric$(\lambda)$.

After a new packet is generated at the beginning of a slot, the transmitter starts transmission of the packet in that slot itself.
For reliable transmission, the $K$-bit packet is assumed to be encoded using a finite-length block code.
The codeword length or transmission duration (in slots) of each packet is controllable.
The decision about the (possibly random) transmission duration of a packet is made at the packet's generation slot.
Thus, the generation slots constitute the decision epochs of the transmitter.
The transmission duration of the $m^{th}$ packet is denoted by $\tau_{m} \ge 1$.
The slot in which the $m$\textsuperscript{th} packet's transmission finishes is denoted as $R[m]$, note that $R[m] = 
T[m] + \tau_{m}$.

Since we consider critical systems where each transmission needs to meet a reliability constraint, we assume that the transmitter power and transmission duration are chosen so as to meet a block codeword error probability constraint.
For a packet encoded and transmitted using a transmission duration of $\tau$ slots let the transmit power be denoted as $P(\tau)$.
We now discuss the tradeoff between $\tau$ and $P(\tau)$ under a reliability constraint.

\subsection{Model for transmit power $P(\tau)$ as a function of $\tau$}
\label{section:tau_vs_p}
We consider a point-to-point link where the transmission duration $\tau$ of a packet can be chosen by the transmitter from $\brac{\tau_{min}, \tau_{min} + 1, \dots, \tau_{max}}$, where $\tau_{min} < \tau_{max} \in \mathbb{Z}_{++}$.
Suppose a packet of length $K$ bits is encoded and transmitted using a codeword with transmission duration $\tau$.
Then, the rate of transmission is denoted as $\rho = K/\tau$.
For the motivating scenarios considered in this paper, short packet communication (SPC) techniques are used
in such scenarios and employ finite blocklength codewords for ensuring reliability of transmission. 
Consider a block coding scheme that transmits at rate $\rho$ over an additive white Gaussian noise (AWGN) channel with received power $P$.
From Polyanskiy's normal approximation \cite{5452208} the optimal codeword length $\tau$ with a codeword error probability 
guarantee of $P_{e,r}$ satisfies
\begin{equation*}
K \approx \tau C_{G} -\sqrt{\tau V_{G}} \mathbb{Q}^{-1}\left(P_{e, r}\right),
\end{equation*}
where $C_{G}$ is the AWGN channel capacity, $V$ is the AWGN channel dispersion\cite{5452208}, and $\mathbb{Q}$ is the Gaussian Q function. 
Using this, we obtain the approximation
\begin{eqnarray}
\tau & = & \bigg\lceil\frac{K}{C_{G}}+\frac{V_{G}\left(\mathbb{Q}^{-1}\left(P_{e, r}\right)\right)^2}{2 C_{G}^2}+ \nonumber \\
& & \frac{\sqrt{V_{G}} \mathbb{Q}^{-1}\left(P_{e, r}\right)}{C_{G}} \sqrt{4 C_{G} K + V_{G}\left(\mathbb{Q}^{-1}\left(P_{e, r}\right)\right)^2}\bigg\rceil.
\label{eq:polyanskiy}
\end{eqnarray}
From \cite{Polyanskiy09dispersionof}, in case of AWGN channel, the channel capacity $C_{G}$ and channel dispersion $V_{G}$ are given by,
\begin{align*}
C_{G} & = \frac{1}{2}\log_{2}(1 + \gamma),\\
V_{G} & = \frac{(\log_2{e})^{2}}{2}\Brap{1 - \frac{1}{(1 + \gamma)^2}},
\end{align*}
where $\gamma = P/N$ denotes the received signal-to-noise ratio, where $P$ is the received power, and $N$ is the noise power.
In this paper, if a reliability of $1 - \epsilon$ is required for each packet, we choose $\tau$ and $P(\tau)$ such that 
$P_{e,r} = \epsilon$.
Also, since the received power is a fraction (pathloss) of the transmit power, we define $P$ to be the transmit power itself.

We illustrate the relationship between $\tau$ and $P(\tau)$ for an example in Figure \ref{fig:ts_vs_power}.
The noise power $N$ is taken to be $0.1$, and the codeword error probability or $\epsilon$ is chosen to be $0.01$.
The parameter $K$ is chosen to be $8$. 
We note that the total energy in a codeword transmission is $P(\tau) \times \tau$ which is also illustrated as a function of $\tau$ in Figure \ref{fig:tau_vs_tau_times_P}.
We note that both $P(\tau)$ and $\tau P(\tau)$ are convex non-increasing functions of $\tau$ and the choice of $\tau$ leads to a tradeoff between transmission duration and transmitter power.

We note that transmission duration and power tradeoffs have been considered in prior work.
For example, Uysal et al. \cite{uysal} consider control of transmission duration but delay rather than age was considered.
Such a tradeoff was also considered in \cite{vineeth_wiopt}, where the tradeoff was characterized using
Shannon's channel capacity theorem for AWGN channels.
The rate of transmission, $\rho$, is given by $\rho = W \log_{2} (1+\gamma)$.
Here $W$ denotes the bandwidth of communication, and $\gamma$ is as defined above. Then,
\begin{equation}
\tau =  \frac{K}{W \log_{2} (1 + P/N)} \text{ and }  P = N\brap{2^{\frac{K}{W\tau}} - 1}.
\label{eq:Shannon}
\end{equation}
Considering $P$ as a function of $\tau$, Shannon-formula based relationship between $\tau$ and $P(\tau)$ is shown in 
Figure \ref{fig:ts_vs_power} and $\tau \times P(\tau)$ as a function of $\tau$ is given in 
Figure \ref{fig:tau_vs_tau_times_P}.
The parameter values are: $K = 800, N = 0.1$, and $W = 50$.
\begin{figure}[!htb]
     \centering
     \begin{subfigure}[b]{0.49\linewidth}
         \centering
         \includegraphics[width=1\linewidth]{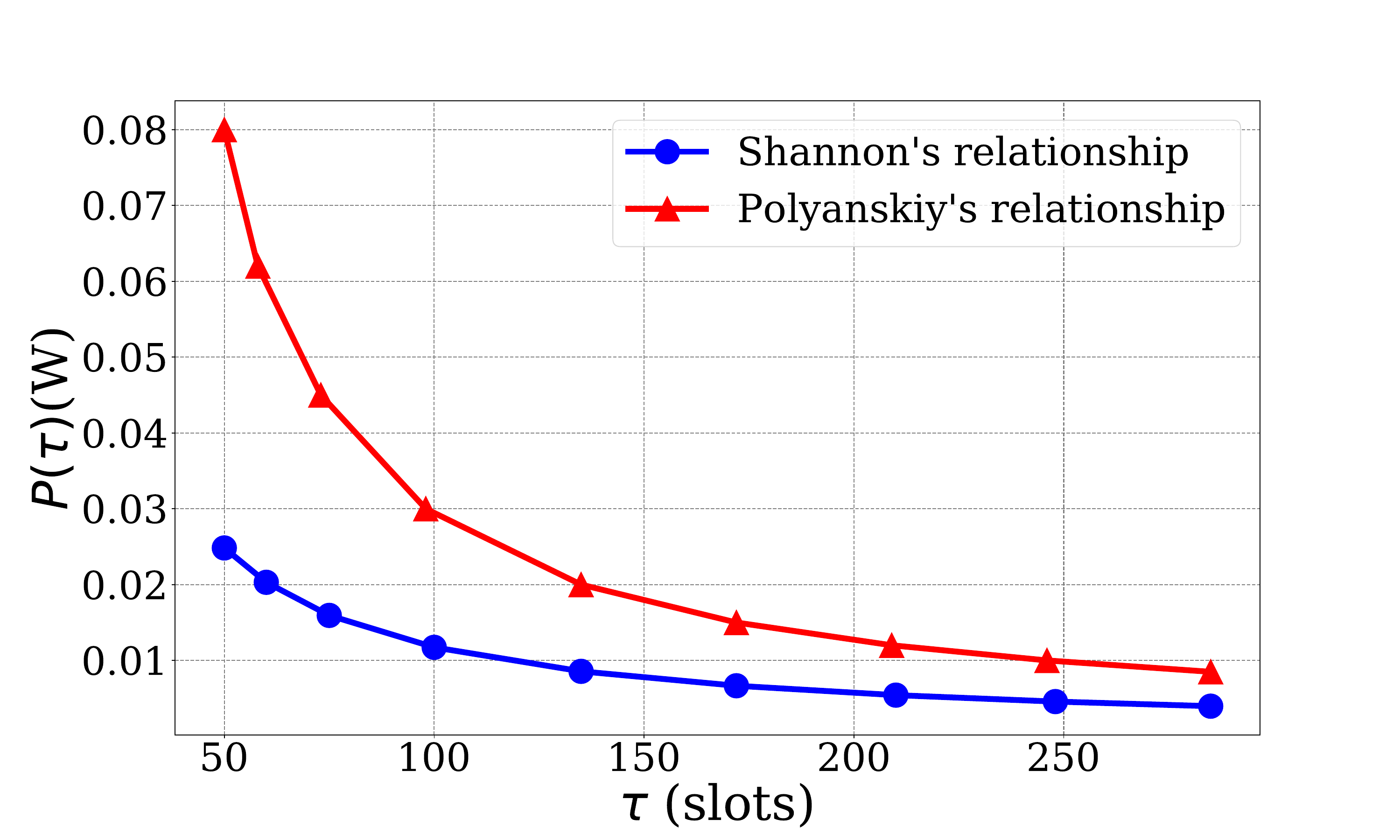}  
         \caption{} \label{fig:ts_vs_power}
     \end{subfigure}
     \hfill
     \begin{subfigure}[b]{0.49\linewidth}
         \centering
         \includegraphics[width=1\linewidth]{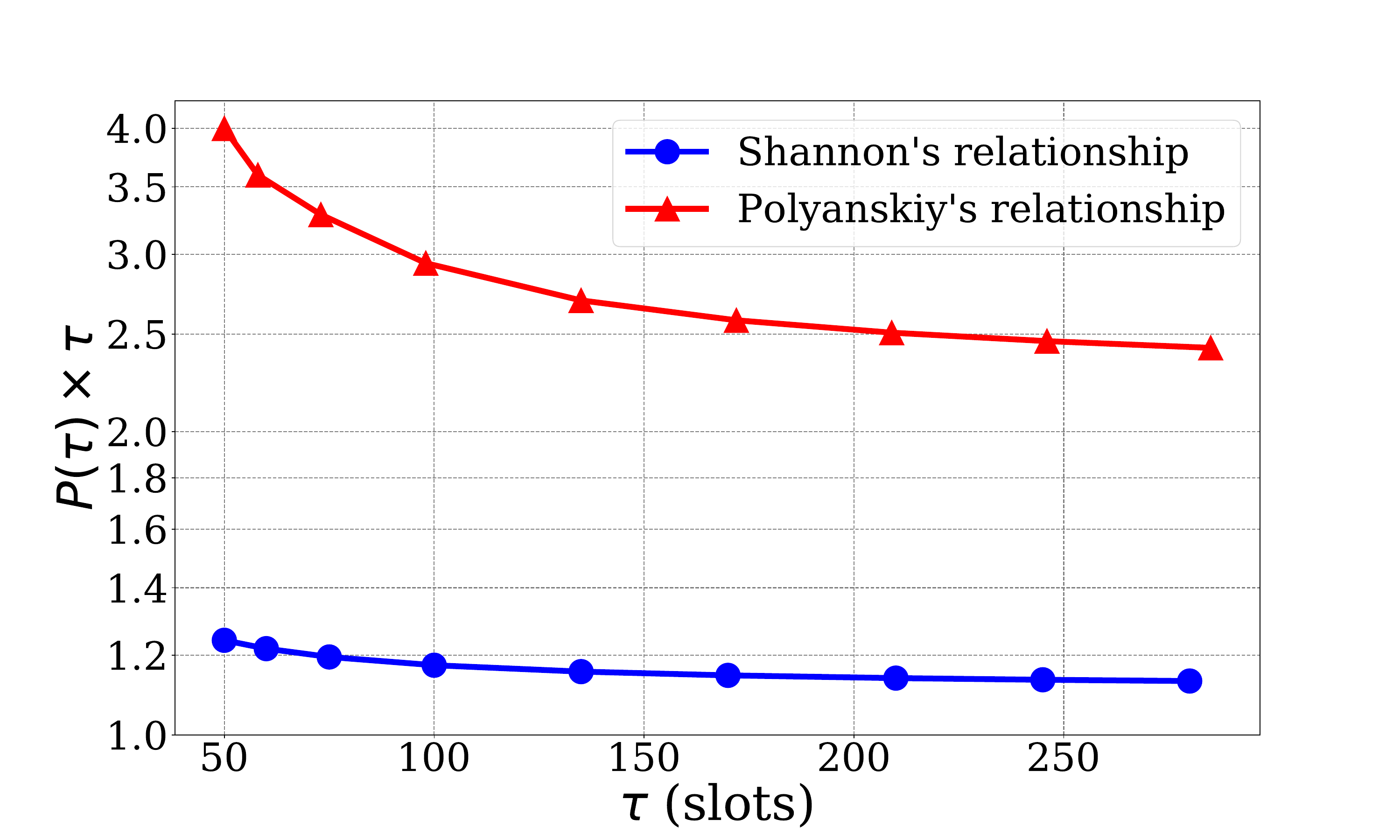}
         \caption{} \label{fig:tau_vs_tau_times_P}
     \end{subfigure}
     \caption{Illustration of the tradeoff between (a) the transmission duration $\tau$ and the transmit power $P(\tau)$ and (b) the transmission duration $\tau$ and the product $P(\tau) \times \tau$ (the total energy in a transmission).}
     \label{fig:convexity}
\end{figure}

\subsection{Transmission policies and the tradeoff problem}
A transmission policy (denoted as $\pi$) decides a transmission duration or codeword length for every packet at its respective decision epoch.
This policy is a (possibly randomized) function of the past evolution of the age of information process of the 
packets, as well as past decisions. 
The AoI process \cite{6195689} (denoted by $A[t], \forall t$) is defined as the time elapsed at the receiver since the 
generation time of the latest successfully received packet.
So, at time slot $t$, if $L[t]$ is the index of the latest successfully received packet, then $T[L[t]]$ is the slot in 
which that packet was generated.
Then, the age
\[ A[t] \triangleq t - T[[L[t]].\]

We note that $A[t]$ drops at $R[m]^{th}$ slot when the $m$\textsuperscript{th} packet is received (i.e., received without error).
A transmission can also result in an error with the packet being not received, in which case the age $A[t]$ would 
increase.
We note that the choice of $\tau$ and $P(\tau)$ for each transmission is such that the probability of receiving the packet is $1 - \epsilon$.
We also note that (from \cite{5452208})
\begin{equation*}
    \varepsilon \approx \mathbb{Q}\left(\frac{\sqrt{\tau}\left(\ln (1+\gamma)-\frac{K}{\tau}\right)}{\sqrt{1-\frac{1}{(1+\gamma)^2}}}\right).
\end{equation*}

Independently of whether a packet was received successfully or not, a new packet would be generated at the transmitter 
according to the process discussed earlier.
At the end of every transmission, the receiver sends a feedback to the transmitter whether the current transmission is successful or not, which enables the transmitter to also compute $A[t]$.
The age at the $m^{th}$ packet's decision epoch is denoted as $A_{m}$, i.e. $A_{m} = A[T[m]]$.

More formally, a transmission policy $\pi$ chooses a transmission duration $\tau_{m}$ for the $m^{th}$ packet at
$T[m]$ as a possibly randomized function \mbox{$\pi(A_{m}, (A[t], t < T[m]), (\tau_{k}, k < m))$}.
The set of all transmission policies is denoted by $\Pi$.
We also consider a class of stationary randomized policies $\Pi_{s}$ that chooses $\tau_{m}$ as a randomized function
$\tau(\cdot)$ of $A_{m}$. For a policy $\pi \in \Pi_{s}$ we define the average age of information (AAoI) as 
\begin{eqnarray*}
\overline{A}^{\pi} = \limsup_{T \rightarrow \infty} \frac{1}{T} \sum_{t = 0}^{T - 1} \Exp A[t].
\end{eqnarray*}
We define $P[t]$ as the transmit power in slot $t$. We note that $P[t] = P(\tau_{m})$ if the $m$\textsuperscript{th}
packet is being transmitted in slot $t$.
Then, for a policy $\pi \in \Pi_{s}$, we define the average power as
\begin{eqnarray*}
\overline{P}^{\pi} = \limsup_{T \rightarrow \infty} \frac{1}{T} \sum_{t = 0}^{T - 1} \Exp P[t].
\end{eqnarray*}

\noindent\textbf{AAoI-Power Tradeoff Problem:}
The AAoI and average power tradeoff problem that we consider in this paper is:
\begin{eqnarray*}
\min_{\pi \in \Pi_{s}} & \overline{A}^{\pi} \\
\text{ s.t. } & \overline{P}^{\pi} \leq p_{c},
\label{eq:tradeoffproblem}
\end{eqnarray*}
where $p_{c} > 0$ is an average power constraint\footnote{This constrained optimization problem, but over $\pi \in \Pi$, can be formulated as a constrained Markov decision process (CMDP) \cite{altman}.
From \cite{altman}, under some technical assumptions, it can be shown that the class of stationary randomized policies contains an optimal policy. This motivates our restriction to $\pi \in \Pi_{s}$ in this paper.}.
The optimal value of the above problem (if it exists) is denoted by $A^{*}(p_{c})$.
In the following sections, we characterize $A^{*}(p_{c})$ analytically and numerically. We note that the Pareto points of the above tradeoff can also be obtained by considering the following unconstrained optimization problem:
\begin{eqnarray}
\min_{\pi \in \Pi_{s}} \overline{A}^{\pi} + \beta \overline{P}^{\pi},
\label{eqn:optimization_SMDP}
\end{eqnarray} 
where the power constraint has been taken into the objective function using the Lagrange approach ($\beta \geq 0$ is a Lagrange multiplier).
We denote this unconstrained version as \textbf{U-AAoI-Power} tradeoff problem.

\section{AAoI-Power tradeoff problem}
In this section, we first characterize the optimal tradeoff using a semi-Markov decision process formulation\footnote{In 
this approach, the policy is obtained numerically for an appropriately state-truncated system. Therefore, it is only 
approximately optimal for the actual system.} with infinite horizon average cost criterion.
We then obtain an analytical upper bound on $A^{*}(p_{c})$ by obtaining $\overline{A}^{\pi}$ and $\overline{P}^{\pi}$ 
for a specific family of transmission policies that uses a fixed transmission duration.
We also obtain lower bounds on $A^{*}(p_{c})$ and show that the above family of fixed transmission duration policies is 
\emph{order-optimal}.

\subsection{Semi-Markov Decision Process Formulation} 
\label{section:SMDP_packet_loss}
A Semi-Markov Decision Process (SMDP) is a valuable tool for analyzing minimum-cost stochastic control problems in which decision epochs occur at random intervals rather than fixed time steps \cite{TijmsH.C2003Afci}.
A discrete space and action SMDP is characterized by the tuple $\brap{\mathcal{S}, \mathcal{A}, 
\mathbb{P},\tilde{\tau}, c}$, where $\mathcal{S}$ is a discrete set of possible states and $\mathcal{A}$ is a 
discrete set of possible actions or decisions. 
The state evolution from decision epoch to the next decision epoch is Markov with state transition probability denoted 
by $\mathbb{P}$.
More precisely, $\mathbb{P}\brap{s^\prime \mid s, u}$ is the probability that at the next decision epoch, the system 
will be in state $s^{\prime}$ if action $u$ is chosen in the present state $s$. 
The expected time until the next decision epoch is $\tilde{\tau}\brap{s,u}$. 
The expected cost at a decision epoch is denoted as $c(s, u)$ if action $u$ is chosen in the state~$s$.

In our problem, the state space $\mathcal{S}$ of the process is the set of all possible age values at a decision 
epoch; i.e., $\mathcal{S}$ is $\sZZ$.
We assume that any transmission duration $\tau$ in the action space $\mathcal{A} = \{ \tau_{min}, \tau_{min} + 1, \cdots,
\tau_{max} \}$ can be chosen. 
Since the transmission durations are discrete-valued, the transmit power $P(\tau)$ (as defined in Section 
\ref{section:tau_vs_p}) takes a set of discrete values in the interval
$\bras{P_{min}, P_{max}}$, where we denote $P(\tau_{min})$ by $P_{min}$ and $P(\tau_{max})$ by $P_{max}$.
Under the assumption that the transmission policy is stationary, the $m^{th}$ packet is transmitted using a transmission duration $\tau_{m} = \tau(A_{m})$.
We note that the decision epochs of the SMDP coincide with the generation times of the packets.
The time duration between $m$\textsuperscript{th} and $(m+1)$\textsuperscript{th} decision epochs is  $\tau_{m} + {\tilde{G}}$, since the $(m + 1)^{th}$ decision epoch is at the first arrival after the $m^{th}$ packet finishes transmission.
The state evolution embedded at decision epochs is a Markov chain. 
For the state transition, two cases arise depending on whether packet error occurs or not.
If a packet error does not occur (with probability $1 - \epsilon$) then after the $m^{th}$ packet finishes 
transmission, the age drops to $\tau(A_{m})$.
Then, the $(m+1)^{th}$ decision epoch occurs after a further $\tilde{G}$ slots so that the age $A_{m + 1} = \tau(A_{m}) 
+ \tilde{G}$.
If packet error occurs, then the age increments by the time between two decision epochs.
Thus, the transition from $A_{m}$ to $A_{m + 1}$ is as follows 
(refer Figure ~\ref{fig:cases_of_aoi_evolution_packet_loss}). 
\begin{figure}[!h]
     \centering
     \begin{subfigure}[b]{0.49\linewidth}
         \centering
         \includegraphics[width=1\linewidth]{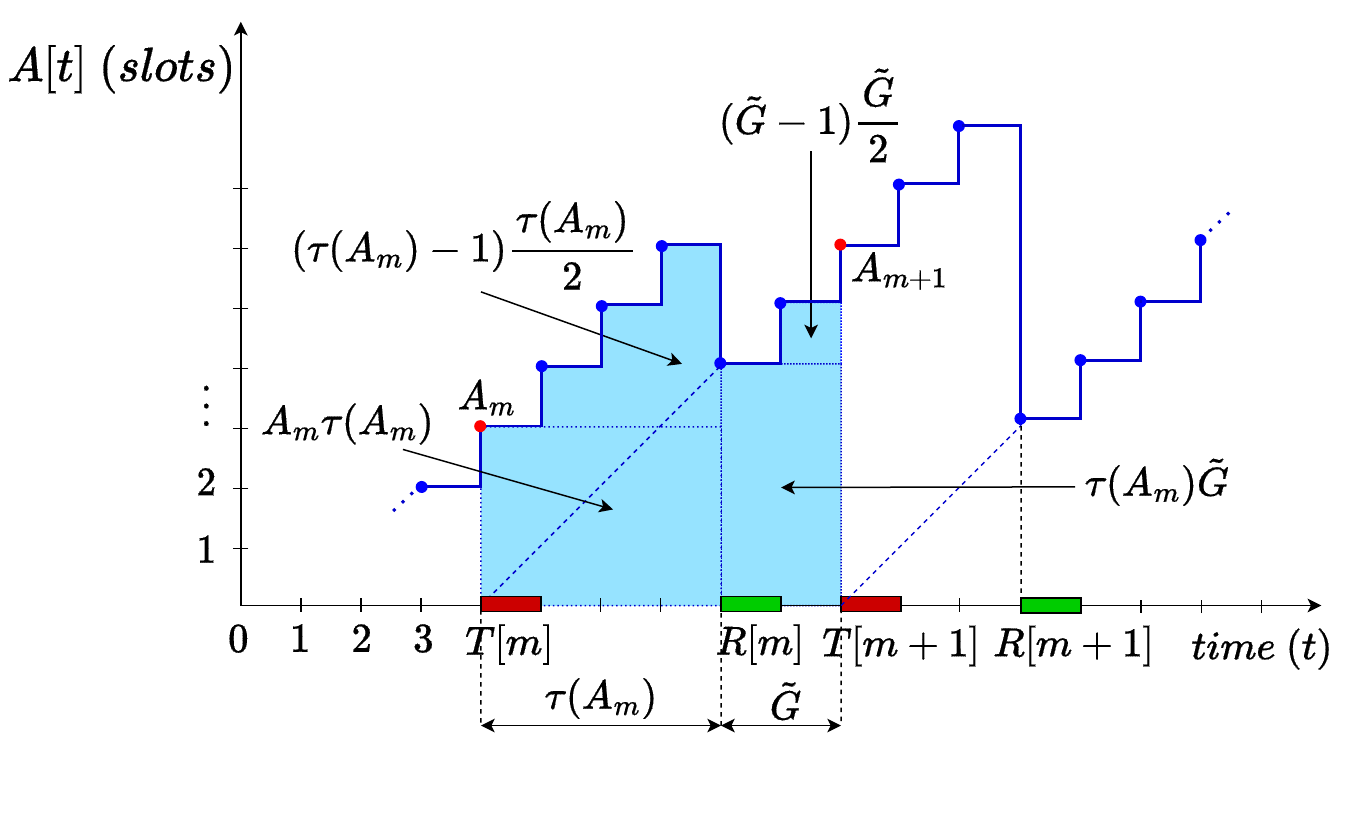}
         \caption{Packet delivers successfully (with probability $1-\varepsilon$).}
     \end{subfigure}
     \hfill
     \begin{subfigure}[b]{0.49\linewidth}
         \centering
         \includegraphics[width=1\linewidth]{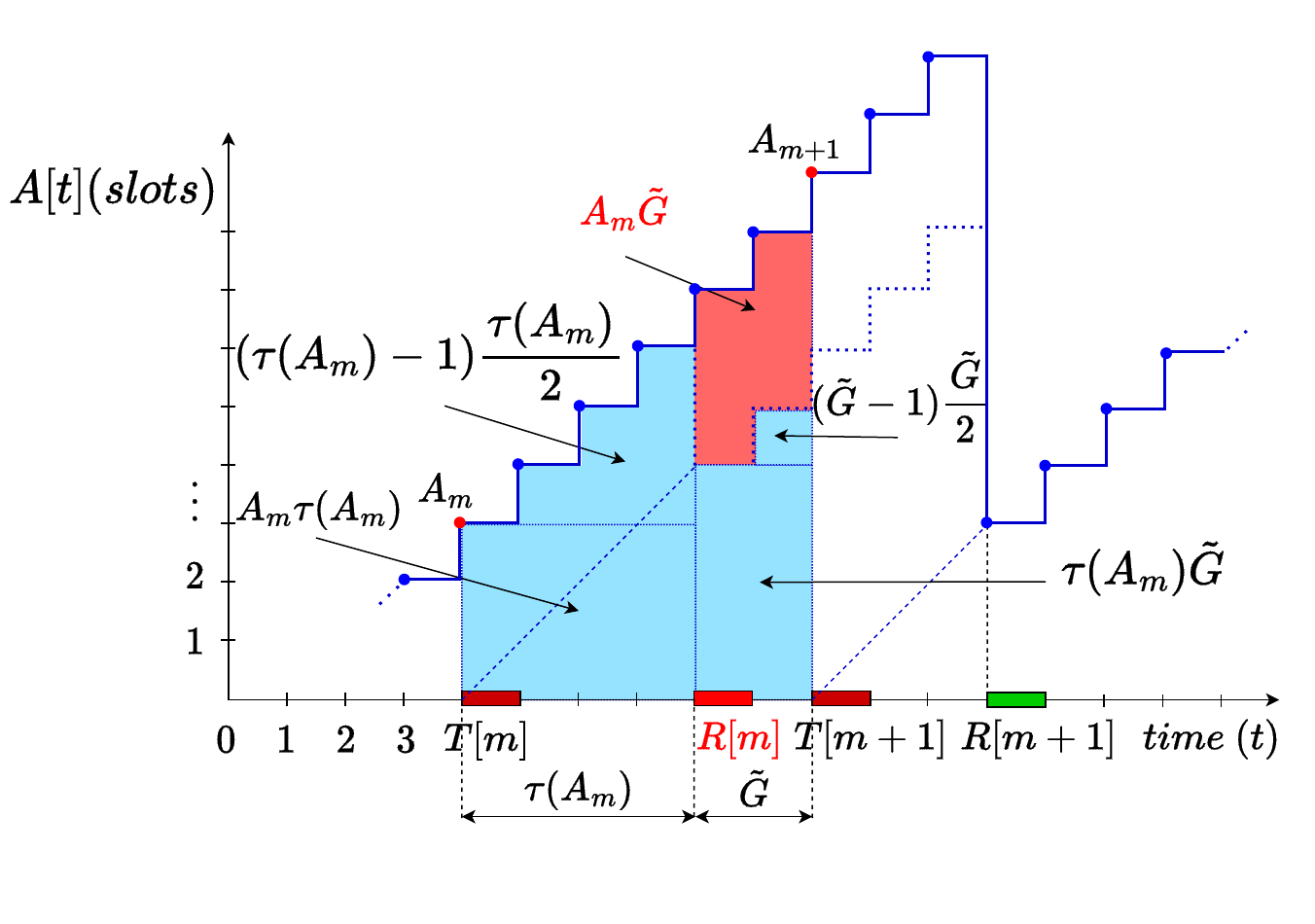}
         \caption{Packet error (with probability $\varepsilon$).}
     \end{subfigure}
     \caption{Illustration of the evolution of $A[t]$ from $A_{m}$ to $A_{m+1}$ depending on whether packet error does not occur (in a) or occurs (in b). The contribution of cumulative age in the single-stage cost is also shown.}
     \label{fig:cases_of_aoi_evolution_packet_loss}
\end{figure}
\begin{equation*}
    A_{m+1} = \begin{cases}
    A_{m}+\tau(A_{m})+\tilde{G} &\text{ with probability } \varepsilon, \\
    \tau(A_{m})+\tilde{G} &\text{ with probability } 1-\varepsilon.
    \end{cases} 
\end{equation*} 
The conditional distribution $\operatorname{Pr}(A_{m+1} = a^{\prime} \mid A_{m} = a, \tau(A_{m}) = \tau)$ is 
\begin{equation*}
    \begin{cases}
    \varepsilon(1 - \lambda)^{a^{\prime}-a-\tau}\lambda + (1-\varepsilon)(1-\lambda)^{a^{\prime}-\tau}\lambda & \text{ for } a^{\prime} \ge a+\tau, \\
    (1-\varepsilon)(1-\lambda)^{a^{\prime}-\tau}\lambda \hspace{2.85 cm} & \text{ for } \tau \le a^{\prime} < a + \tau, \\
    0 \hspace{5.65 cm} & \text{ for } a^{\prime} < \tau,
    \end{cases} 
\end{equation*}
where we have used that $\tilde{G} \sim$ Geometric($\lambda$).
We consider the objective function in \eqref{eqn:optimization_SMDP} for the SMDP.
To minimize this function, we define the following single stage cost $c(a, \tau)$, which is the expected cumulative age 
and power over the time duration between two consecutive decision epochs. Here $a$ is the age  at the decision epoch, 
and $\tau$ is the service time. 
\begin{align*}
c(a, \tau) =~& a\tau + (\tau-1)\frac{\tau}{2} + \tau \frac{1-\lambda}{\lambda} + \Brap{\frac{1-\lambda}{\lambda}}^{2}\\&+\varepsilon a\frac{1-\lambda}{\lambda} + \beta P(\tau) \tau \label{eq:lr_ssc_np}
\end{align*}
The different components of the expected cumulative age can be seen in Figure \ref{fig:cases_of_aoi_evolution_packet_loss}.

We use the value iteration algorithm\cite{TijmsH.C2003Afci} to arrive at the optimal policy for a truncated version of 
this SMDP denoted by $\pi_{SMDP}$. 
The average AoI and power for $\pi_{SMDP}$ (denoted by $\overline{A}^{\pi_{SMDP}}$ and $\overline{P}^{\pi_{SMDP}}$ 
respectively) provides a baseline which can be used to evaluate the tradeoff performance of other policies.

\subsection{An upper bound on the tradeoff}\label{sec:upper_bound}
We obtain an upper bound on the tradeoff by analytically characterizing the averages $\overline{A}^{\pi}$ and 
$\overline{P}^{\pi}$ for a family of policies called fixed transmission duration policies.

\noindent\textbf{Fixed transmission duration policy $\pi_{t_{s}}$:}
A fixed transmission duration (FTT) policy transmits every packet in $t_{s}$ slots with power $P(t_{s})$ for a fixed error 
probability of $\epsilon$.
The parameter $t_{s}$ can be changed to obtain different $\overline{A}^{\pi}$ and $\overline{P}^{\pi}$. 
A small $t_{s}$ is expected to give a combination of \emph{large} $\overline{P}^{\pi}$ and \emph{small} 
$\overline{A}^{\pi}$ compared to a large~$t_{s}$. 
FTT policies can be used to achieve the end points of the AAoI-Power tradeoff as shown in the following proposition.
\begin{proposition}
An FTT policy with $t_{s} = \tau_{max}$ ($t_{s} = \tau_{min}$) is optimal at the minimum power (maximum power)
end point of the AAoI-Power tradeoff. 
\end{proposition}
\begin{proof}
Consider any stationary policy $\pi$ with a stationary distribution on $A_{m}$.
Since $\tau_{m} = \tau(A_{m})$ we obtain an induced stationary distribution on $\tau_{m}$.
Using Markov Renewal Reward Theorem (MRRT) \cite[Appendix D]{kumar2008wireless}, the 
average power for a stationary policy $\pi$ is
\begin{eqnarray*}
 \overline{P}^{\pi} = \frac{\Exp P(\tau)\tau}{\Exp \tau + (1 - \lambda)/\lambda},
\end{eqnarray*}
where the expectation is with respect to the above stationary distribution of $\tau$.
We note that
\begin{eqnarray*}
 \frac{\Exp P(\tau)\tau}{\tau_{max} + (1 - \lambda)/\lambda} \leq \overline{P}^{\pi} \leq \frac{\Exp P(\tau)\tau}{
\tau_{min} + (1 - \lambda)/\lambda},
\end{eqnarray*}
for any $\pi$.
Then using the monotonic decreasing property of $P(\tau)\tau$ we observe that 
\begin{eqnarray*}
 \frac{P(\tau_{max})\tau_{max}}{\tau_{max} + (1 - \lambda)/\lambda} \leq \overline{P}^{\pi} \leq \frac{ 
P(\tau_{min})\tau_{min}}{
\tau_{min} + (1 - \lambda)/\lambda}.
\end{eqnarray*}
The lower and upper bounds on the average power are achieved iff the transmission duration is $\tau_{max}$ and 
$\tau_{min}$ respectively.
Therefore, the only feasible policies at the minimum power and maximum power end points of the tradeoff are FTT 
policies with $t_{s} = \tau_{max}$ and $t_{s} = \tau_{min}$ respectively.
\end{proof}

\noindent We analytically characterize the AAoI and average power for FTT policies in the following proposition.
\begin{proposition} 
\label{proposition_packet_loss}
For FTT policy with $t_{s}$ and $P(t_{s})$ chosen such that packet error probability is $\epsilon$, the AAoI 
$\overline{A}^{\pi_{t_{s}}}$ is
\begin{equation*}
\overline{A}^{\pi_{t_{s}}}=t_{s}+\frac{\mathbb{E}R^{2}}{2 \mathbb{E}R}-\frac{1}{2}, \label{eq:aaoi_loss}
\end{equation*}
and the average power $\overline{P}^{\pi_{t_{s}}}$ is 
\begin{equation*}
    \overline{P}^{\pi_{t_{s}}} = \frac{P(t_{s})t_{s}\lambda}{1-\lambda + \lambda t_{s}},
\label{eq:pavg_loss}
\end{equation*} 
where  
    \begin{align*}
    \Exp R =~& \frac{1}{1-\varepsilon}\Brap{\frac{1-\lambda}{\lambda} + t_{s}}, \\
    \Exp R^{2} =~& \frac{1-\lambda}{(1-\varepsilon)\lambda^{2}}+\frac{1+\varepsilon}{(1-\varepsilon)^{2}}\Bras{\frac{1-\lambda}{\lambda}+t_{s}}^2.
    \end{align*}
\end{proposition} 
The derivation of these expressions is discussed in Appendix~\ref{proof_proposition_packet_loss}.
We note that as $\lambda \downarrow 0$, the AAoI behaviour is $\mathcal{O}\brap{\frac{1}{\lambda}}$ due to the scarcity 
of packets being generated.
As expected, we observe that $\overline{P}^{\pi_{t_{s}}}$ is monotonically decreasing in $t_{s}$ while 
$\overline{A}^{\pi_{t_{s}}}$ is monotonically increasing in $t_{s}$.
We note that the tradeoff performance of FTT policies (obtained by varying $t_{s}$) provides an analytical upper bound 
to the AAoI-Power tradeoff.
The above analytical characterization helps in designing an FTT policy to achieve a particular point in the AAoI power 
tradeoff plane.

We also consider a generalization of FTT policies which is the family of threshold policies.

\noindent\textbf{Threshold policy $\pi_{h}$:}\label{section:threshold_Pre}
A threshold policy is parameterized by a  threshold $h$ on age at a decision epoch and two transmission durations 
$\tau_{a}$ 
and $\tau_{b}$ ($\tau_{a}, \tau_{b} \in \mathcal{A}$ with $\tau_{a} > \tau_{b}$).
The threshold policy chooses the transmission duration as a function $\tau(A_{m})$ of the age at a decision epoch.
The function 
\begin{equation}
    \tau(A_{m}) = 
    \begin{cases}
    \tau_{a} \text{ if } A_{m} \leq h, \\
    \tau_{b} \text{ if } A_{m} > h.
    \end{cases}
\end{equation}
We note that when $h$ is small, the policy uses the smaller transmission duration $\tau_{b}$ to transmit the packets most 
of the time (i.e., 
unless the age is below $h$ at the decision epoch); this comes at the cost of a higher average power consumption. When 
$h$ is large, it uses the larger service time $\tau_{a}$ most of the time; this lowers the average power consumption 
but 
could lead to a large average age. Thus, by varying the threshold $h$, as well as $\tau_{a}$ and $\tau_{b}$, we obtain 
a tradeoff between AAoI and average power, which is an upper bound to the AAoI-Power tradeoff\footnote{We 
present an analytical characterization of the tradeoff for threshold policies in an error-free system in 
Section \ref{sec:errorfree_system}. This characterization is used for obtaining candidate parameter values for the 
threshold policies which is then used in simulations to compute the tradeoff.}.

In the next section, we obtain lower bounds to the optimal tradeoff $A^{*}(p_{c})$ which can be used to investigate 
the usefulness of the heuristic policies.

\subsection{Lower bounds on AAoI-Power tradeoff}
\label{section:NP}
The system discussed above is denoted as $\sysA$.
In order to obtain a lower bound on $A^{*}(p_{c})$ for $\sysA$, we construct another system denoted as $\mathbb{B}$.
We assume that $\sysB$ consists of two independent point-to-point links: (1) the point-to-point link of $\sysA$ and 
(2) a point-to-point link which has the same properties as that of the first link, but is error-free.
The packet generation process is assumed to be same for both $\sysA$ and $\sysB$.
In $\sysB$, a generated packet and its copy are transmitted on the first and second links.
The two links use the same transmission duration $\tau_{m}$ at a packet generation instant $m$.
Thus, the decision epochs for the two links would coincide.

The transmission duration $\tau_{m}$ decisions are made as follows.
We define an additional function $A_{B}[t]$ for the second link in $\sysB$ as follows:
\[ A_{B}[t] = t - T[L_{B}[t]], \]
where $L_{B}[t]$ is the index of last packet received for the second link of $\sysB$ by slot $t$.
We note that $L_{B}[t]$ is the index of the last packet which \emph{could have been} received by 
slot $t$ in the first link if all transmissions were successful.
The earlier definition of $A[t]$ for $\sysA$ is retained for the first link.
Similar definitions as in Section \ref{section:SMDP_packet_loss} can be made for $\sysB$.
We consider the state of $\sysB$ to be the tuple $(A_{m}, A_{B,m})$ where $A_{B,m} = A_{B}[T[m]]$.

A stationary policy for $\sysB$ chooses a transmission duration as a possibly randomized function 
$\tau(A_{m}, A_{B,m})$. 
We note that any stationary policy $\pi \in \Pi_{s}$ for $\sysA$ can be implemented for $\sysB$ by neglecting $A_{B,m}$\footnote{We also
note that a policy for $\sysB$ which neglects $A_{m}$ corresponds to a policy in an error-free system where the age 
evolution is $A_{B}[t]$. An error-free system would be a natural choice for a \emph{dominated} system, the 
construction of which is an usual technique to obtain a lower bound. However, decision 
epochs may occur at different times in $\sysA$ and the error-free system. This motivated the construction of $\sysB$ 
which couples the evolution of $\sysA$ and an error-free system together.}.
Under $\pi$, the transmission durations $\tau_{m} = \tau(A_{m})$ are chosen for both the first and second links.
We note that after every $m^{th}$ transmission, $A_{B}[t]$ drops to $\tau_{m}$.
However, $A[t]$ may or not may not reduce depending on whether the packet transmission is successful.
If a packet transmission is successful, then we note that at the end of that packet transmission duration $A[t] = 
A_{B}[t]$.
Therefore, $A_{B}[t] \leq A[t]$.
We also note that $\forall t$, $P[t]$ is the same for the first and second link of $\sysB$.

We now present a numerical lower bound on the tradeoff by considering system $\sysB$.
We note that the set of all stationary policies for $\sysA$ is a subset of the set of all stationary policies for 
$\sysB$.
Therefore, a lower bound on the time average of $A_{B}[t]$ for $\sysB$ considering the set of all stationary policies 
is a lower bound on the time average of $A_{B}[t]$ for $\sysB$ considering only those stationary policies from $\sysA$.
Since $A_{B}[t] \leq A[t]$, we then obtain a lower bound on $\overline{A}^{\pi}$ using the lower bound on the time 
average of $A_{B}[t]$ over all policies.
A lower bound on the time average of $A_{B}[t]$ is obtained in the following proposition.

\begin{proposition}
The minimum average AoI for $\sysA$ over all policies in $\Pi_{s}$ with a power constraint $p_{c}$, i.e. $A^{*}(p_{c})$, is bounded from below by the optimal value of the following optimization problem:
 \begin{align} \begin{split}
    \min_{p(\tau, \tau^{\prime})} &\frac{\sum\limits_{\tau}\sum\limits_{\tau^{\prime}}p(\tau, 
\tau^{\prime})\Bras{\tau(\Exp\tilde{G} + \tau^{\prime})+\frac{1}{2}(\Exp\tilde{G} + \tau^{\prime})(\Exp\tilde{G} + 
\tau^{\prime}-1)}}{\sum_{\tau}\sum_{\tau^{\prime}}p(\tau,\tau^{\prime})\tau^{\prime}+\Exp\tilde{G}}, \\ 
    \text{s. t.} & \sum_{\tau}\sum_{\tau^{\prime}}p(\tau, 
\tau^{\prime})\Brap{P(\tau^{\prime})\tau^{\prime}-p_{c}\tau^{\prime}} \le p_{c}\Exp\tilde{G}, \\
    & p(\tau,\tau^{\prime}) \ge 0, \\
    &\sum_{\tau}\sum_{\tau^{\prime}}p(\tau, \tau^{\prime}) = 1.
    \end{split}
    \label{eq:optimization_LB}
\end{align}
Here $p_{c}$ is the average power constraint and $\Exp\tilde{G}$ is $\frac{1 - \lambda}{\lambda}$ (the mean of the time 
duration for generating a packet after a packet transmission is completed).
The optimization is done over the variables $p(\tau, \tau^{\prime})$, where $\tau, \tau' \in \brac{\tau_{min}, \dots, 
\tau_{max}}$.
The variables $p(\tau, \tau')$ are interpreted as the joint probability of two consecutive transmission 
times being $\tau$ and $\tau'$.
\label{prop:numerical_lb}
\end{proposition}
We note that \eqref{eq:optimization_LB} is a linear fractional program in $p_{\tau,\tau'}$ and can be solved using the 
Charnes-Cooper transformation. 
We denote this numerical lower bound as $A_{n}(p_c)$. 

\begin{proof}
We first obtain the time average $\overline{A}_{B}^{\pi}$ of $A_{B}[t]$ and average power using Markov renewal 
reward theorem.
For this, we identify a semi-Markov process (SMP) in the evolution of $\sysB$ for any stationary policy $\pi$ as well 
as appropriate costs for a renewal cycle.

At every packet arrival epoch $T[m]$ consider $(A_{m}, A_{B,m})$.
The evolution of $(A_{m}, A_{B,m})$ constitutes a Markov chain, which is the embedded Markov chain (EMC) of the SMP 
under $\pi$ for $\sysB$.
\newcommand{\poltauB}{\tau(A_{m},A_{B,m})}
The duration of time between successive epochs of the EMC is $\tau(A_{m}, A_{B,m}) + \tilde{G}$.
The EMC and the inter-epoch durations define the SMP.

We associate two cumulative costs with the SMP over each transition.
The cumulative age $c((A_{m}, A_{B,m}), \tau)$ cost is 
\begin{align*}
\mathbb{E} & \biggl[A_{B,m} \tau + \frac{\tau\left(\tau-1\right)}{2} + 
\tau\tilde{G}\\& + \frac{\tilde{G}(\tilde{G}-1)}{2}\biggr],
\end{align*}
where we have used $\tau = \poltauB$.
The cumulative power cost is $\mathbb{E}\left[P(\tau)\tau\right]$.

We note that using the Markov renewal reward theorem, we obtain the time average $\overline{A}_{B}^{\pi}$ of $A_{B}[t]$ 
and average power as 
$$\frac{\Exp\bras{c(A_{m}, A_{B,m}, \tau)}}{\Exp\tau + \frac{1-\lambda}{\lambda}} \text{ and } 
\frac{\Exp\bras{P(\tau) \tau}}{\Exp\tau + \frac{1-\lambda}{\lambda}},$$ respectively,
where the expectation is with respect to the stationary distribution (assumed to exist) of the $(A_{m}, A_{B,m})$ EMC.
The average $\overline{A}_{B}^{\pi}$ is a lower bound for $\overline{A}^{\pi}$ since $A_{B,m} \leq A_{m}$ 
implies that $c(A_{m}, A_{B,m}, \tau) \leq 
c(A_{m}, \tau)$.

For $\sysB$, $A_{B,m} \sim \tilde{G} + \tau_{m - 1}$ since $A_{B}[t]$ evolves under the assumption that the 
transmissions are error-free.
Consider the cost $c((A_{m}, A_{B,m}), \tau)$.
We note that the first term $A_{B}\tau$ can be bounded from below by $(\tilde{G}+\tau_{m - 1})\times\tau_{m}$.

If we assume that the EMC reaches a steady state with a stationary distribution under $\pi$, then there is a 
corresponding stationary joint distribution for $(A_{m - 1}, A_{B,m - 1})$ and $(A_{m}, A_{B,m})$.
For a stationary policy $\pi$, this induces a stationary joint distribution of $\tau_{m - 1}$ and $\tau_{m}$.
We denote this stationary joint distribution as $p(\tau, \tau')$ for $\tau_{m - 1} = \tau$ and $\tau_{m} = \tau'$.
Then, using MRRT, we can write the average power $\overline{P}^{\pi}$ as 

\begin{equation*}
    \overline{P}^{\pi} = 
\frac{\sum_{\tau}\pi_{\tau}\sum_{\tau^{\prime}}p(\tau,\tau^{\prime})P(\tau^{\prime})\tau^{\prime}}{\sum_{\tau}\sum_{
\tau^{\prime}}\pi_{\tau} p(\tau,\tau^{\prime})\tau^{\prime} + \Exp\tilde{G}} = \frac{\sum_{\tau'}\pi_{\tau'} 
P(\tau')\tau'}{\sum_{\tau'}\pi_{\tau'}\tau' + \Exp\tilde{G}}. 
\end{equation*}
The lower bound on the average age $\overline{A}_{B}^{\pi}$ as
\begin{equation*}
    = \frac{\sum\limits_{\tau}\sum\limits_{\tau^{\prime}}p(\tau, 
\tau^{\prime})\Bras{\tau'(\Exp\tilde{G} + 
\tau)+\frac{\tau^{{\prime}^{2}}}{2}+\tau^{\prime}\Exp\tilde{G}+\frac{\Exp\tilde{G}^{2}}{2}-\frac{\Exp\tilde{G}
+\tau^{\prime}}{2}}}{\sum_{\tau}\sum_{\tau^{\prime}}p(\tau,\tau^{\prime})\tau^{\prime}+\Exp\tilde{G}}.  
\end{equation*}

Therefore for the average power constraint $p_{c}$, the optimization problem can be expressed as 
\begin{align} \begin{split}
    \min_{p(\tau, \tau^{\prime})} &\frac{\sum\limits_{\tau}\sum\limits_{\tau^{\prime}}p(\tau, 
\tau^{\prime})\Bras{\tau'(\Exp\tilde{G} + \tau)+\frac{1}{2}(\Exp\tilde{G} + \tau^{\prime})(\Exp\tilde{G} + 
\tau^{\prime}-1)}}{\sum_{\tau}\sum_{\tau^{\prime}}p(\tau,\tau^{\prime})\tau^{\prime}+\Exp\tilde{G}}, \\ 
    \text{s. t.} & \sum_{\tau}\sum_{\tau^{\prime}}p(\tau, 
\tau^{\prime})\Brap{P(\tau^{\prime})\tau^{\prime}-p_{c}\tau^{\prime}} \le p_{c}\Exp\tilde{G}, \\
    & p(\tau,\tau^{\prime}) \ge 0, \\
    &\sum_{\tau}\sum_{\tau^{\prime}}p(\tau, \tau^{\prime}) = 1.
    \end{split}
\end{align}
\end{proof}

\noindent Using a similar approach as above, we obtain the following analytical lower bound on the AAoI-Power tradeoff.
\begin{proposition}
The minimum average AoI for $\sysA$ for all policies in $\Pi_{s}$ with a power constraint $p_{c}$, i.e. $A^{*}(p_{c})$, is bounded from below as
follows:
\begin{eqnarray*}
    A^*(p_{c}) \geq \frac{c_{l}(\tau^{*})}{\tau_{max} + (1 - \lambda)/\lambda}, 
\end{eqnarray*}
where $\tau^{*}$ is the smallest real-valued $\tau \in [\tau_{min}, \tau_{max}]$ such that $\frac{\lambda\tau 
P(\tau)}{1-\lambda + 
\lambda \tau} \leq p_{c}$ and
\begin{align*}
c_{l}(\tau) = &\biggl[2\tau\frac{1-\lambda}{\lambda}+\tau \tau_{\text {min 
}}\\&+\frac{\tau(\tau-1)}{2}+\Brap{\frac{1-\lambda}{\lambda}}^{2}\biggr].
\end{align*}
\label{proposition_NP_LB} 
\end{proposition}
\begin{proof}
As in the proof of Proposition \ref{prop:numerical_lb} we consider the SMP in the evolution of $\sysB$ with EMC 
$(A_{m}, A_{B,m})$.
Similar to the proof above we define the cumulative age reward, denoted by $c(A_{m}, A_{B,m}, \tau)$, as
\begin{align*}
\mathbb{E} & \biggl[A_{B,m} \tau + \frac{\tau\left(\tau-1\right)}{2} + 
\tau\tilde{G}\\& + \frac{\tilde{G}(\tilde{G}-1)}{2}\biggr].
\end{align*}
We note that using the Markov renewal reward theorem, we obtain the average AoI and average power as 
$$\frac{\Exp\bras{c(A_{m}, A_{B,m}, \tau)}}{\Exp\tau + \frac{1-\lambda}{\lambda}} \text{ and } 
\frac{\Exp\bras{P(\tau) \tau}}{\Exp\tau + \frac{1-\lambda}{\lambda}},$$ respectively,
where the expectation is with respect to the stationary distribution of the $(A_{m}, A_{B,m})$ EMC.

From the data transformation method \cite{TijmsH.C2003Afci} we construct another Markov chain where the transitions are 
of unit slot duration, with single stage age and power rewards as $\frac{c(A_{m}, A_{B,m}, \tau)}{\tau + \frac{1 - 
\lambda}{\lambda}}$ and $\frac{P(\tau)\tau}{\tau + \frac{1 - \lambda}{\lambda}}$ such that the average AoI and average 
power can be written as
\begin{eqnarray*}
\tilde{\Exp}\bras{\frac{c(A_{m}, A_{B,m}, \tau)}{\tau + \frac{1-\lambda}{\lambda}}} \text{ and } 
\tilde{\Exp}\bras{\frac{P(\tau)\tau}{\tau + \frac{1 - \lambda}{\lambda}}},
\end{eqnarray*}
respectively.
Here the expectations (denoted as $\tilde{E})$) are with respect to the stationary distribution of the data transformed 
Markov chain.
The average AoI and powers for the SMP and the data transformed Markov chain are the same for a policy $\pi$.

Let us denote the stationary version of $A_{m}$ by $A$ and $A_{B,m}$ by $A_{B}$.
We now consider the problem
\begin{eqnarray}
\text{minimize} & &  \tilde{\Exp}\bras{\frac{c(A, A_{B}, \tau)}{\tau + \frac{1-\lambda}{\lambda}}} \nonumber \\
\text{such that} & & \tilde{\Exp}\bras{\frac{P(\tau)\tau}{\tau + \frac{1 - \lambda}{\lambda}}} \leq p_{c}.
\label{eq:datatransformed_lb}
\end{eqnarray}
To obtain a lower bound on the above optimization problem\footnote{We note that the same lower bound can be obtained 
by considering $\sysA$ and bounding $A_{m}$ from below by $(\tilde{G}+\tau_{min})$.}, we bound the first term $A_{B} 
\tau$ in $c(A, 
A_{B}, \tau)$ from below by $(\tilde{G}+\tau_{min})\times\tau$, since $A_{B} \geq \tilde{G}+\tau_{min}$.
This modified cost denoted as $c_{l}(\tau)$ is then a function only of $\tau$.
We also bound the denominator term $\tau + (1 - \lambda)/\lambda$ from above using $\tau_{max} + (1 - \lambda)/\lambda$.
We also note that the average power is also a function of $\tau$.
Thus, we have the following optimization problem, where we optimize over all possible choices of the distribution of a 
random variable $\tau \in [\tau_{min}, \tau_{max}]$. Note that we have relaxed the integer constraint on $\tau$, which 
is allowed since we seek a lower bound.
\begin{eqnarray*}
\text{minimize} & &  \Exp\bras{\frac{{c_{l}(\tau)}}{\tau_{max} + \frac{1-\lambda}{\lambda}}}, \\
\text{such that} & & \Exp\bras{\frac{{P(\tau) \tau}}{\tau + \frac{1-\lambda}{\lambda}}} \leq p_{c}.
\end{eqnarray*}
The optimal value of the above problem is a lower bound to \eqref{eq:datatransformed_lb}.
We note the objective function  $\frac{c_{l}(\tau)}{\tau_{max} + \frac{1-\lambda}{\lambda}}$ is a convex increasing 
function in $\tau$, while the constraint function is convex decreasing in $\tau$. Therefore, by Jensen's inequality, an 
optimal distribution would assign probability only to a single value $\tau^{*}$ 
which is the smallest $
\tau^{*}$ such that the constraint is satisfied.
The approximate lower bound is then $\frac{{c_{l}(\tau^{*})}}{\tau_{max} + \frac{1-\lambda}{\lambda}}$, where $\tau^*$ 
is the smallest $\tau \in [\tau_{min}, \tau_{max}]$ such that $\frac{\lambda\tau P(\tau)}{1-\lambda + 
\lambda \tau} \leq p_{c}$.
\end{proof}
\begin{remark}
We note that with the modified age cost function, a single transmission duration $\tau^{*}$ is optimal. 
This raises a question of whether FTT policies would have optimal or near-optimal performance for the actual tradeoff 
problem.
\end{remark}

\subsection{Optimality properties of FTT policies}
In this section, we discuss some optimality properties of FTT policies.
We first consider an approximate error-free system to model the original system evolution for high reliability (i.e. $\epsilon \approx 0$).
We assume that packets are always received without error but the $\tau-P(\tau)$ relationship is for the non-zero but small $\epsilon$.
We then have the following proposition.
\begin{proposition}
 FTT policies are AAoI-power tradeoff optimal for the error-free system if $\lambda = 1$.
 \label{prop:FTToptimal_highlambda_lowepsilon}
\end{proposition}
\begin{proof}
We construct another Markov chain using the data transformation method \cite{TijmsH.C2003Afci} where the transitions 
are of unit slot duration and the single stage age and power costs are $\frac{c(A_{m}, \tau)}{\tau + \frac{1 - 
\lambda}{\lambda}}$ and $\frac{P(\tau)\tau}{\tau + \frac{1 - \lambda}{\lambda}}$.
Since $\lambda = 1$, the AAoI and average power are
\begin{eqnarray*}
\tilde{\Exp}\bras{\frac{c(A_{m}, \tau)}{\tau }} \text{ and } 
\tilde{\Exp}\bras{\frac{P(\tau)\tau}{\tau }},
\end{eqnarray*}
respectively.
Here the expectations (denoted as $\tilde{E})$) are with respect to the stationary distribution of the data 
transformed Markov chain.

For $\lambda = 1$, the function $c(a, \tau)/\tau$ is jointly convex in $a$ and $\tau$.
Thus, a distribution that gives unit mass to a single value for $(a, \tau)$ would achieve the minimum value for AAoI.
For a FTT policy with parameter $t_{s}$, if $\lambda = 1$ and $\epsilon = 0$, the stationary distribution gives unit mass to the value $(t_{s}, t_{s})$ for $(a, \tau)$.
Thus, FTT policies are optimal.
\end{proof}

\begin{remark}
 From Proposition \ref{prop:FTToptimal_highlambda_lowepsilon} we expect that FTT policies would be close to optimal for $\lambda \approx 1$ and $\epsilon \approx 0$.
 Non-adaptive policies may therefore be appropriate for optimally trading off AAoI and power in this high packet-generation rate and high reliability regime.
\end{remark}

We now show that FTT policies (with time-sharing) also satisfy a weaker form of optimality (defined as order-optimality in the following) for the error-free system with $\lambda < 1$.
For obtaining this property, we first derive an analytical (but asymptotic) lower bound on the AAoI as $p_{c} \downarrow P_{min}$.

\begin{proposition}
Let $\delta_{k} \downarrow 0$ be a monotonically decreasing sequence of positive real numbers.
Then, for any sequence of stationary policies $\pi_{k} \in \Pi_{s}$ such that $\overline{P}^{\pi_{k}} \leq P_{min} + \delta_{k}$ and sufficiently small $\delta_{k}$ we have that
\begin{eqnarray*}
 \overline{A}^{\pi_{\tau_{max}}} - \overline{A}^{\pi_{k}} = \mathcal{O}(\delta_{k}),
\end{eqnarray*}
for the error-free system.
\label{prop:lowpower_lb}
\end{proposition}
The proof is given in Appendix \ref{app:proof:proposition_orderoptimality}.
The asymptotic lower bound is obtained by bounding the probability of using transmission durations other than $\tau_{max}$ by $\mathcal{O}(\delta_{k})$ terms.

We define a notion of \emph{order-optimality} for a family of policies.
Consider the AAoI-Power tradeoff problem with a sequence of $p_{c,k} = P_{min} + \delta_{k}$. 
From the above proposition we have that $A^{*}(P_{min}) - A^*(p_{c,k})$ for every $p_{c,k} \leq P_{min} + \delta_{k}$ is 
$\mathcal{O}(\delta_{k})$.
We define a family of policies (such as FTT) to be $\delta$ order-optimal if there is a sequence of policies $\pi_{p_{k}}$ (with parameter(s) $p_{k}$) with $\overline{P}^{\pi_{p_{k}}} \leq P_{min} + \delta_{k}$ such that $A^{*}(P_{min}) - \overline{A}^{\pi_{p_{k}}}$ is $\Omega(\delta_{k})$ as $\delta_{k} \downarrow 0$.

We have that the FTT family of policies is order-optimal for the error-free system.
\begin{proposition}
The family of FTT policies (with time-sharing) is $\delta$ order-optimal for the error-free system.
\end{proposition}
\begin{proof}
 We construct a sequence $\pi_{k}$ of time shared FTT policies to show order-optimality.
 We note that the set of unique FTT policies is obtained by choosing the parameter $t_{s}$ from $\brac{\tau_{1} = \tau_{min}, 
\dots,\tau_{k} = \tau_{max}}$.
 From Proposition \ref{proposition_packet_loss} we have that $\overline{A}^{\tau}$ is monotonically increasing in $\tau$ 
while $\overline{P}^{\tau}$ is monotonically decreasing in $\tau$.
 Consider $p_{c} = P_{min} + \delta_{k}$ where $\delta_{k}$ is small.
 We can define a parameter $\alpha_{k}$ such that 
 \begin{equation*}
  \alpha_{k} P_{min} + (1 - \alpha_{k}) \overline{P}^{\pi_{t_{s}}} = P_{min} + \delta_{k}.
 \end{equation*}
 Here $t_{s} < \tau_{max}$ and $\alpha_{k}$ is a time-sharing parameter (the two FTT policies $\pi_{\tau_{max}}$ and $\pi_{t_{s}}$ are time-shared in the proportion $\alpha_{k}$ to $1 - \alpha_{k}$ with the duration of time used for a policy tending to infinity).
 We then note that 
 \begin{equation*}
  \alpha_{k} = 1 - \frac{\delta_{k}}{\overline{P}^{\pi_{t_{s}}} - P_{min}}.
 \end{equation*}
 The AAoI of this time shared policy $\pi_{k}$ is $\overline{A}^{\pi_{p_{k}}} = \alpha_{k} A^{*}(P_{min}) + (1 - \alpha_{k}) \overline{A}^{\pi_{t_{s}}}$ so that as $\delta_{k} \downarrow 0$, $A^{*}(P_{min}) - \overline{A}^{\pi_{p_{k}}}$ is $\Omega(\delta_{k})$.
\end{proof}

\begin{remark}
For the AAoI-Power tradeoff problem, the above result motivates the use of FTT policies in a regime where the power constraint $p_{c}$ approaches $P_{min}$.  
We note the following technical points: 
\begin{enumerate}
 \item The class of time-shared policies is not stationary. However, we can define a larger class of quasi-stationary policies, where we assume that there is a quasi-stationary state which is chosen with some probability for a sample evolution of the system. The asymptotic lower bound derived in Proposition \ref{prop:lowpower_lb} can be extended to the case of quasi-stationary policies (by using a joint distribution over the quasi-stationary state and two consecutive transmission durations in the proof of Proposition \ref{prop:lowpower_lb}).
 \item If we restrict to the class of stationary policies, then a family of state-independent randomized policies can be shown to be order-optimal. A policy in this family chooses a transmission duration at random independent of the state (age) at a decision epoch. Such policies can also be shown to be order-optimal. We note that in many other scenarios (such as those considered in \cite{neely}) state-independent randomized policies are not order-optimal.
\end{enumerate}
\end{remark}

\begin{remark}
 We note that a similar order-optimality result also holds for a \emph{high-power} regime.
 The maximum power $P_{max}$ is obtained by using the FTT policy with $t_{s} = \tau_{min}$.
 If we consider the AAoI-Power tradeoff problem for $p_{c} \geq P_{max} - \delta$, then we obtain that 
$A^{*}(p_{c}) = \overline{A}^{\pi_{\tau_{min}}} + \mathcal{O}(\delta)$.
 The result holds since similar $\mathcal{O}(\delta)$ bounds on the stationary probability of using a transmission 
duration $\tau \neq \tau_{min}$ can be obtained.
 These bounds are outlined in Appendix \ref{sec:app:bounds_highpower}.
\end{remark}

In summary, non-adaptive FTT policies (or state-independent randomized policies) are \emph{good candidate policies} for achieving the optimal AAoI-Power tradeoff for high-reliability systems.

\subsection{Numerical \& Simulation Results}
\label{low_reliability_results}
In this section, we evaluate the performance of the family of FTT policies.
We compare the AAoI-Power tradeoff achieved by the families of FTT policies and threshold policies with that achieved by $\pi_{SMDP}$ as well as the lower bounds.

The AAoI-Power tradeoff is illustrated in Figures \ref{fig:tradeoffs_packet_loss} and 
\ref{fig:all_in_one_NP} for error probabilities: (a) $\epsilon = 0.01$ and (b) $\epsilon = 0.2$.
In Figure \ref{fig:comparison_policies} the arrival rate $\lambda = 0.1$. The minimum and maximum transmission durations 
$\tau_{min} = 24$ and $\tau_{max} = 138$ are chosen. Corresponding power levels $P(\tau_{min}) = 10$~mW and 
$P(\tau_{max}) = 1$~mW from \eqref{eq:polyanskiy} for noise power $N = 10$~mW, packet length 
$K = 8$ bits.
For obtaining the tradeoff under the family of FTT policies, we vary the parameter $t_s \in \brac{\tau_{min}, \dots, 
\tau_{max}}$ and plot the AAoI and average power from Proposition~\ref{proposition_packet_loss}.

\begin{figure}[!htb]
    \centering
    \begin{subfigure}{0.49\linewidth}
    \includegraphics[width = 1\linewidth]{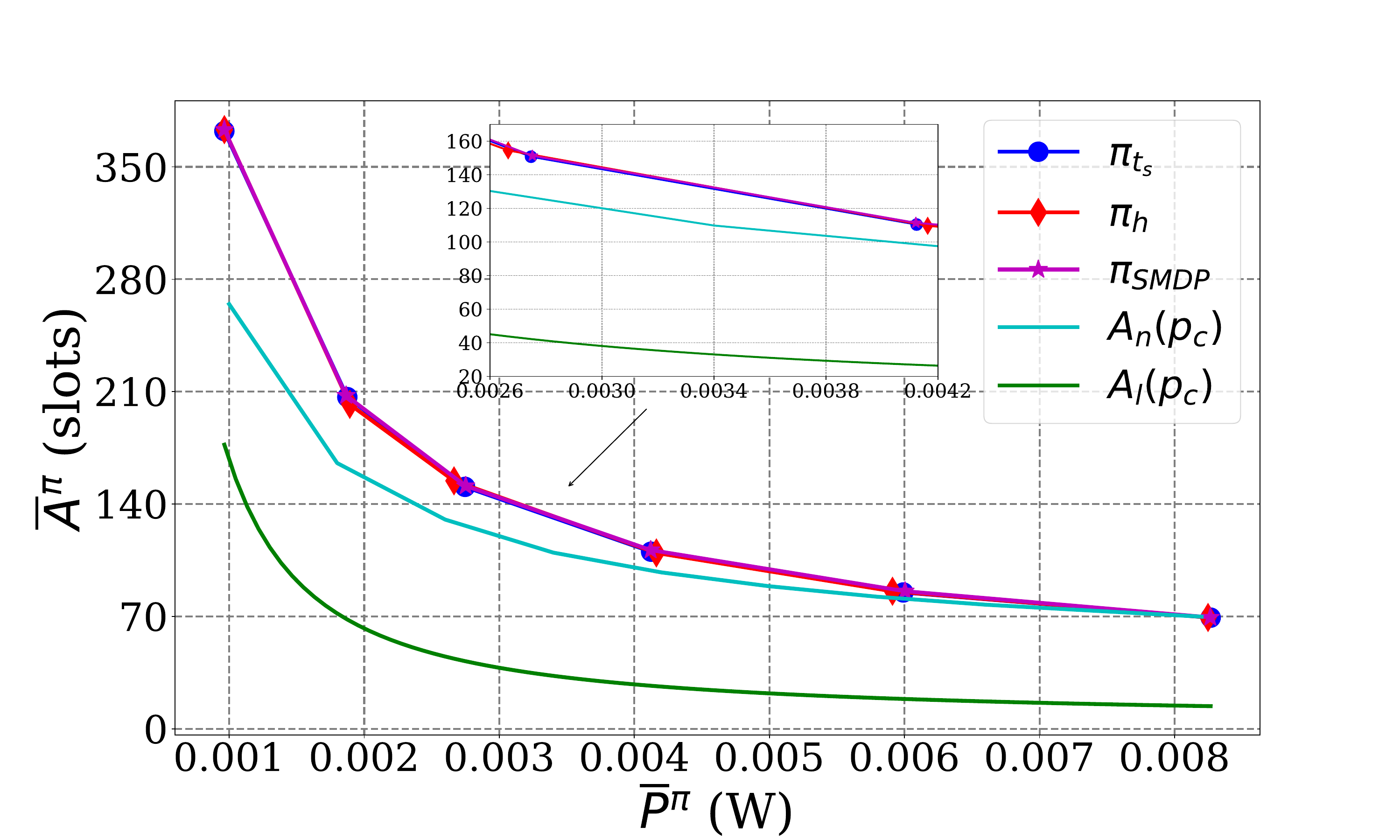}
    \caption{$\epsilon = 0.01$}
    \label{fig:all_in_one_NP}
    \end{subfigure}
    \begin{subfigure}{0.49\linewidth}
    \includegraphics[width = 1\linewidth]{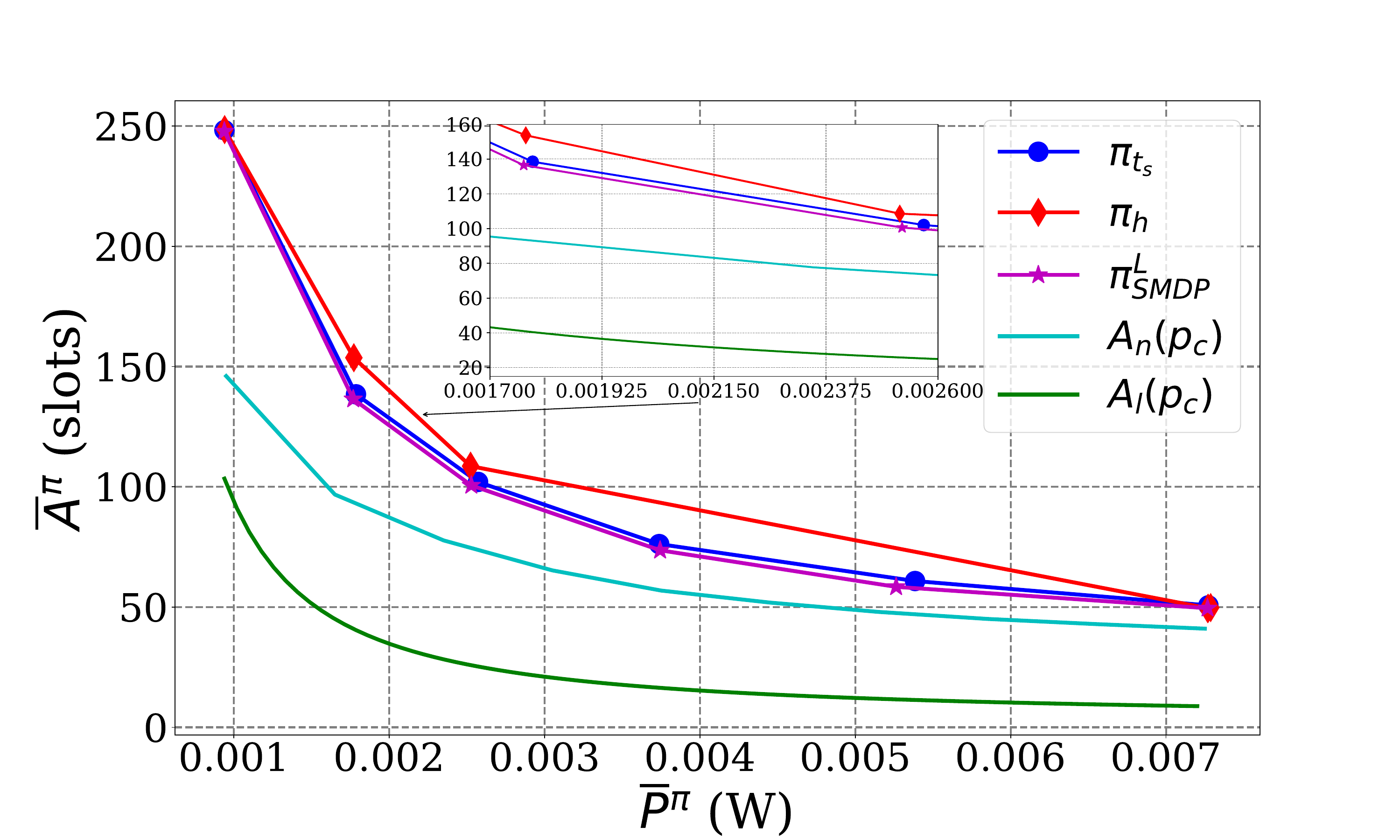}
    \caption{$\epsilon = 0.2$}
    \label{fig:tradeoffs_packet_loss}
    \end{subfigure}
    \caption{Comparison of AAoI-average power tradeoff for FTT policy, threshold policy, and the optimal policy (from SMDP). The numerical lower bound $A_{n}(p_{c})$ as well as the analytical lower bound $A_{l}(p_{c})$ are also shown.}
    \label{fig:comparison_policies}
\end{figure}

The results demonstrate that the proposed family of FTT policies is approximately optimal with respect to the tradeoff performance.
We plot the optimal policies $\pi_{SMDP}$ (for different $\beta$ values) for $\epsilon = 0.01$ and $\lambda = 0.1$ in Figure \ref{fig:SMDP_optimalpolicies}.
For $\beta = 0$ and large $\beta$ ($10^{6}$) we obtain the end points of the AAoI-Power tradeoff for which FTT policies are optimal.
For intermediate values of $\beta$, we observe that FTT is only approximately optimal.

\begin{figure}
 \centering
 \includegraphics[width = 0.5\linewidth]{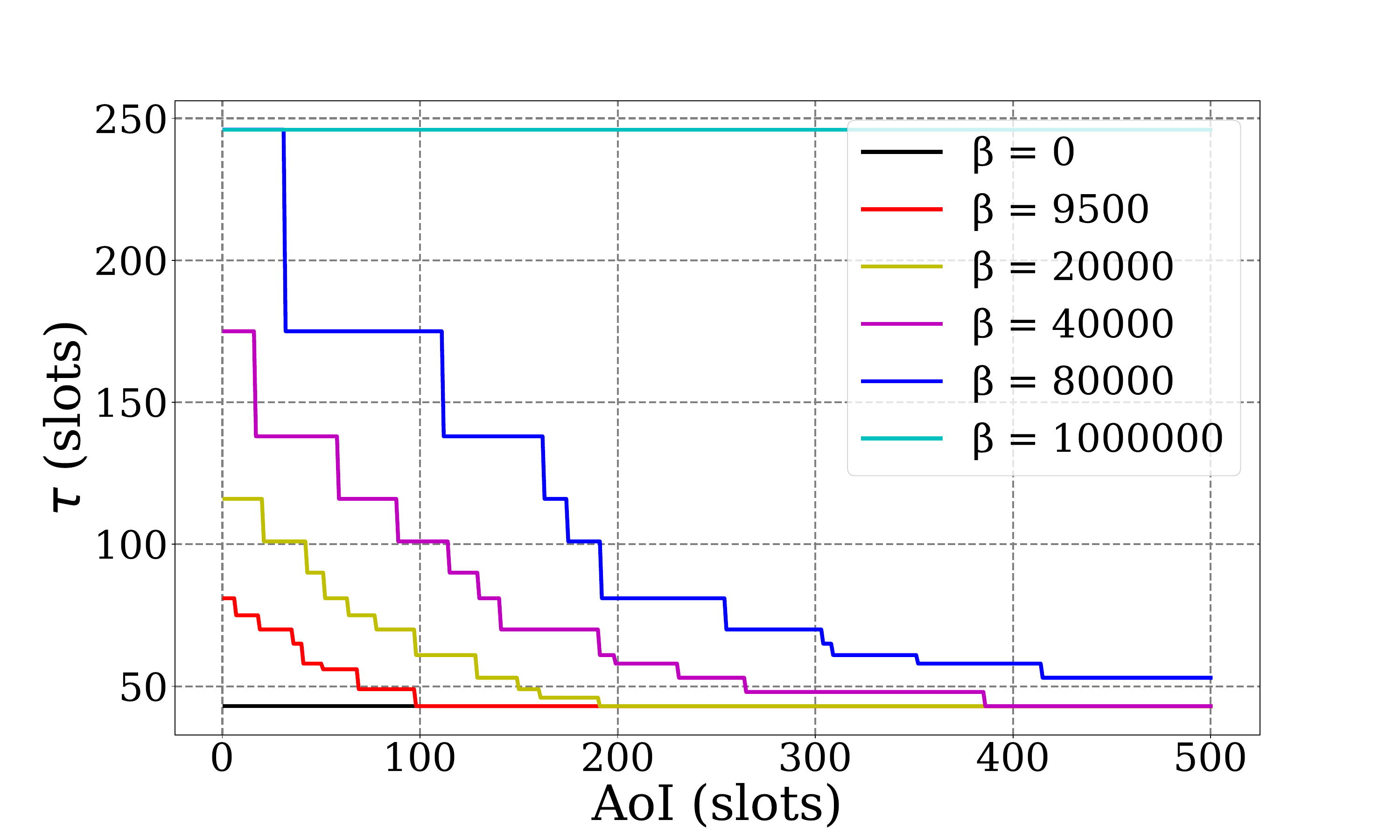}
 \caption{Illustration of the action $\tau(.)$ as a function of the state (age) for optimal policies $\pi_{SMDP}$ for different $\beta$ values. The optimal policy is not FTT except at the end-points of the AAoI-Power tradeoff. The maximum and minimum average power end-points are obtained for $\beta = 0$ and $\beta = 10^{6}$ respectively.}
 \label{fig:SMDP_optimalpolicies}
\end{figure}

We observe that this performance is obtained even for $\epsilon = 0.2$ and $\lambda = 0.1$ which may not be considered to be a high-reliability and large packet generation rate regime.

The analytical lower bound is loose, while the numerical lower bound $A_{n}(p_{c})$ is tight for higher power values and smaller $\epsilon$.
The numerical lower bound is appropriate for performance evaluation in a high power low error probability regime and along with the analytical upper bound for FTT policies provides a way to characterize the tradeoff without computing the optimal policy using the SMDP.

The tradeoff plots for the threshold policy has been obtained using the following steps. 
In the first step, we minimize $\overline{A}^{\pi_{h}} + \beta \overline{P}^{\pi_{h}}$ for different positive 
values of $\beta$.
For a particular $\beta$, the minimization is carried out over the parametes $h$, $\tau_{a}$, and $\tau_{b}$ using an 
analytical characterization of $\overline{A}^{\pi_{h}}$ and $\overline{P}^{\pi_{h}}$ obtained using an error-free 
system.
This analysis is presented in Appendix \ref{sec:errorfree_system}.
We note that the minimization is carried out using an evolutionary algorithm called Differential Evolution which 
yields a local minimum.
In the second step, the locally optimal values of the parameters: threshold $h$ and the transmission durations $\tau_{a}$ and 
$\tau_{b}$ are used to simulate the threshold policy for a system with errors in order to obtain the actual
$\overline{A}^{\pi_{h}}$ and $\overline{P}^{\pi_{h}}$ which are then plotted.
For the case with $\epsilon = 0.01$, the error-free system analysis yields \emph{good} approximations for 
$\overline{A}^{\pi_{h}}$ and $\overline{P}^{\pi_{h}}$ which leads to a choice of parameters that yield near-optimal 
tradeoff performance for the family of threshold policies.
However, for larger $\epsilon = 0.2$, $\overline{A}^{\pi_{h}}$ and $\overline{P}^{\pi_{h}}$ are not approximated well 
enough, so that the optimization does not lead to parameter choices that yield near-optimal tradeoff performance.


\section{Age-Power Tradeoff for Block-Fading Channels}
In this section, we consider a block-fading point-to-point channel model and evaluate the AAoI-Power tradeoff for FTT policies, which were found to have approximately optimal performance.
We consider a block-fading model where the channel is assumed to be constant for a block of $T$ consecutive symbols and then changes \cite{marzetta1999capacity}.
This block of $T$ consecutive symbols is the coherence time of the channel.
The channel can be considered approximately constant within each
block, allowing the receiver to estimate the channel and decode the transmitted data using standard techniques. However, 
at the end of each block, a new channel realization is encountered, requiring the receiver to re-estimate the channel 
and adapt its decoding strategy.

A codeword of length $\tau = LT$ spans $L$ independent channel realizations. 
If the receiver has access to channel state information (CSI) but not the transmitter, the maximum achievable rate 
$\rho^{\ast}_{csi}(\tau, \epsilon)$ is asymptotically defined \cite{yang2012diversity, polyanskiy2011scalar}:
\begin{equation}\label{eq:max_rate}
\rho^{\ast}_{csi}(\tau, \epsilon) = 
C_{csi}-\sqrt{\frac{V_{csi}}{\tau}}\mathbb{Q}^{-1}(\epsilon)+o\left(\frac{1}{\sqrt{\tau}}\right),
\end{equation}
where $C_{csi} = \mathbb{E}_{H}[\log(1 + \gamma|H|^{2})]$ is the channel capacity (where $|H|^{2}$ is the random channel 
gain and $\gamma$ is the signal-to-noise ratio (SNR) at the transmitter), $V_{csi} = T \operatorname{Var}[\log(1 + 
\gamma|H|^{2})] + 1 - \mathbb{E}^{2}\left[\frac{1}{1 + \gamma|H|^{2}}\right]$ is the channel dispersion, and $f(x) = 
o(g(x))$ means that $\lim_{x\rightarrow\infty} |f(x)/g(x)| = 0$.
The error probability is $\epsilon$.

To understand the effect of fading on the AAoI-Power tradeoff, we consider a system model which is the same as that considered in Section \ref{section:System_Model} except that the set of possible transmission durations are now multiples of the coherence time, which is $T$ slots.
Essentially, the control of the transmission duration is via the choice of $L \in \brac{1,2,3,\dots}$.
With this choice for the transmission durations, the relationship between SNR at the transmitter ($\gamma$) (or the transmit power $P(\tau)$ for a fixed noise power) for satisfying the error probability of $\epsilon$, and the codeword length $\tau$ in Rayleigh-fading block-fading channel is shown in Figure~\ref{fig:snr_vs_tau_vs_T}.
We have illustrated this relationship for different values of $T$ and an error probability requirement $\epsilon = 0.01$.
The parameter $K = 8$ and $N = 0.01$.
For each value of $T$, the relationship is obtained from \eqref{eq:max_rate} and expressions for $C_{csi}$ and $V_{csi}$.
The corresponding AAoI-power tradeoff for the family of FTT policies is shown in Figure \ref{fig:fading_tradeoff_FTT} where the parameter $\lambda$ is $0.5$.
\begin{figure}[!htb]
    \centering
    \begin{subfigure}{0.49\linewidth}
    \centering
    \includegraphics[width = 1\linewidth]{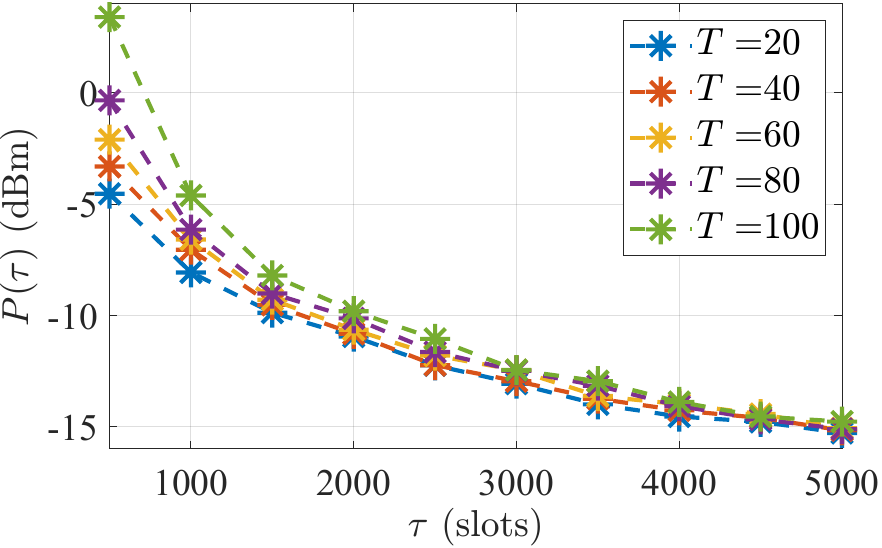}
    \caption{$P(\tau)$ versus $\tau$ for different $T$.}
    \label{fig:snr_vs_tau_vs_T}
    \end{subfigure}
    \begin{subfigure}{0.49\linewidth}
    \includegraphics[width = 1\linewidth]{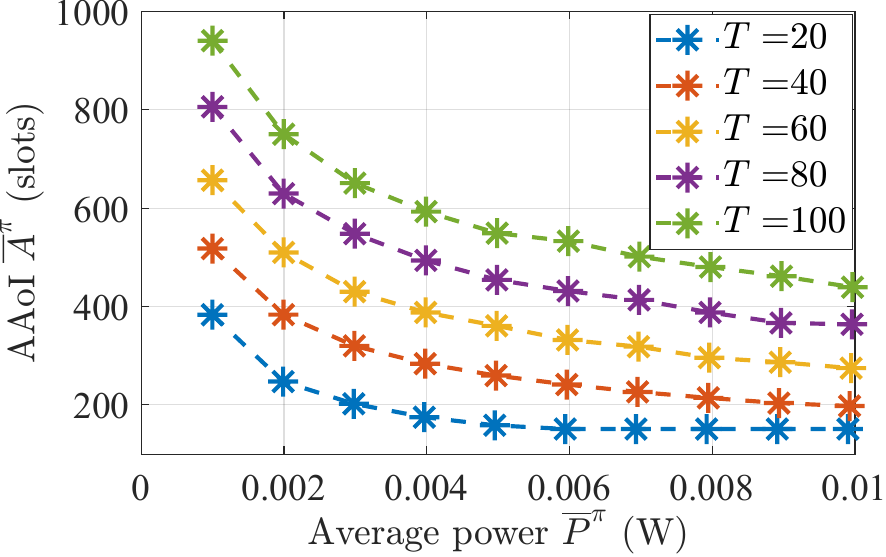}
    \caption{ $\overline{P}^{\pi}$ versus $\overline{A}^{\pi}$ for FTT policy.}
    \label{fig:fading_tradeoff_FTT}
    \end{subfigure}
    \caption{Illustrations of the transmit power-transmission duration tradeoff and AAoI-Power tradeoff for block fading channels with coherence time $T$. The tradeoffs are illustrated for different choices of $T$.}
    \label{fig:fading_tradeoffs}
\end{figure}
We observe that a shorter coherence time $T$, resulting in faster channel dynamics, offers a better tradeoff between the 
SNR (and hence the transmit power) and transmission duration $\tau$ when the receiver has access to the CSI. The intuition 
is that we get a larger channel diversity gain as $L$ increases when coherence time $T$ is small for the same 
block-length $\tau$. Furthermore, it is interesting to note that the SNR or transmit power is a non-increasing convex 
function of the codeword length or transmission duration $\tau$, even in the case of block-fading channels. This implies 
that the AoI-transmit power tradeoff analysis we carried out for the non-fading case can be easily extended to the 
fading channels.

\section{Age-Power tradeoff with other packet generation models}
\label{section:Pre}
The generation rate of update packets plays a crucial role in minimizing AoI.
Different packet generation models have been considered in literature.
Models in which packets are generated independently of the age-evolution process were considered in work such as \cite{6195689} and \cite{kaul2012status}.
Another common packet generation model is the zero-wait model in which a new packet is generated at the source whenever a \emph{server} (that models the system which transfers the packet to the destination) is free.
Such models were considered in \cite{sun2019sampling} and \cite{sun2017update} where it was also shown that the zero-wait scheme is not optimal with respect to minimizing AAoI in the scenarios considered.
Zero-wait with an additional delay was found to perform better.
We note that the model in Section \ref{section:System_Model} considered a packet generation model that generated a packet once the current packet had finished transmission, but with an additional random delay (the parameter $\lambda$ controls the delay in generation of a new packet).
We also note that there are other packet generation models such as periodic packet generation or age-threshold packet generation considered in literature \cite{sun2019sampling}.
In order to better understand the impact of packet generation on the age-power tradeoff, we discuss two packet generation models here.
We characterize the AAoI-Power tradeoff for FTT policies for these packet generation models and compare with the AAoI-Power tradeoff for the model in Section \ref{section:System_Model}.

\noindent\textbf{Independent packet generation process:}
We assume that packets are generated in each slot according to an independent and identically distributed (IID) Bernoulli$(\lambda)$ process.
We note that this models a scenario where packets are generated independently of the age evolution process in the system and is similar to periodic packet generation except that the periods or inter-generation times are random rather than deterministic.

At a packet generation epoch, which corresponded with the decision epoch in our earlier model, we now consider the following cases.
If there is no on-going packet transmission, then the packet generation epoch is a decision epoch at which we choose a packet transmission duration for the generated packet.
However, if there is an on-going packet transmission, then we first have the choice of either pre-empting the current packet transmission or discarding the generated packet.
We note that discarding the generated packet and allowing the current packet transmission to continue till completion is the model that we have considered in Section \ref{section:System_Model}.
The time till the next packet generation is Geometric due to the memoryless property of the inter-generation times.
We consider the alternate model, where at a packet generation time, the current packet transmission is pre-empted and discarded.
A new packet transmission starts for the generated packet after the current packet is discarded.
We call this the pre-emptive (or P) model, while the earlier model in Section \ref{section:System_Model} is called non-preemptive (NP).

\noindent\textbf{Age-threshold based packet generation:}
In this packet generation model, we assume that the transmitter is aware of the age process $A[t]$ at the receiver using error-free feedback.
A new packet is generated in a slot $t$ if there is no on-going packet transmission and if the current age $A[t]$ is greater than or equal to a threshold $h_{a}$.
We call this the age-threshold (AT) model.

\subsection{AAoI-Power tradeoff for FTT policies}
We note that intuitively the P model behaves similarly to the NP model in the regime where $\tau_{max} < \frac{1}{\lambda}$, i.e., for $\lambda \approx 0$.
In this regime, the probability of pre-emption is small.
A detailed analysis of the AAoI-Power tradeoff for the P model including a SMDP formulation to characterize the optimal tradeoff and a lower bound were presented in our prior work \cite{sudarsanan2023optimal} under the assumption of error-free transmissions.

In this section, we analytically characterize the AAoI-Power tradeoff for FTT policies for P and AT packet generation models with packet errors.
We then compare these with the AAoI-Power tradeoff obtained earlier for NP model.
We first consider the P model under a FTT policy with transmission duration of $t_{s}$.
We note that at a packet generation instant, if there is no ongoing transmission then a new packet transmission starts with a duration of $t_{s}$.
If there is an ongoing transmission then that is discarded and a new packet transmission starts with a duration of $t_{s}$.
The following proposition characterizes $\overline{A}^{\pi_{t_{s}}}$ and $\overline{P}^{\pi_{t_{s}}}$ with a packet error probability of $\epsilon$.

\begin{proposition}
For P packet generation model with FTT policy of packet transmission duration $t_{s}$ and packet error probability $\epsilon$ we have that
\begin{eqnarray*}
 \overline{A}^{\pi_{t_{s}}} = \frac{1}{\alpha\lambda} \text{ and } \overline{P}^{\pi_{t_{s}}} = P(t_{s})\brap{1 
- (1 - \lambda)^{t_{s}}}.
\end{eqnarray*}
where $\alpha = (1 - \epsilon)(1 - \lambda)^{t_{s} - 1}$.
\label{prop:Pmodel_tradeoff}
\end{proposition}
The proof of this proposition is given in Appendix \ref{appendix:proofPmodel_tradeoff}.
We observe that the AAoI exhibits a $\frac{1}{\lambda}$ behaviour as $\lambda \downarrow 0$.
Also, since we have preemption of an existing transmission, the AAoI increases to infinity as $\lambda \uparrow 1$.
We also note that for a given average power $\lambda$ and $t_{s}$ could be optimized to obtain the minimum AAoI.

We now consider the AT model with threshold $h_{a}$ and packet transmission duration of $t_{s}$.
We note that if $h_{a} \leq t_{s}$ and $\epsilon < 1$, after sufficiently large time, the AT model behaves as a zero-wait system.
After every packet transmission, which may or may not be in error, the age would be more than $h_{a}$ (the minimum age is $t_{s}$).
So, a new packet would be immediately generated.
If $h_{a} > t_{s}$, then in case the age after a transmission is less than $h_{a}$ there is a delay to the next packet generation epoch.
All transmissions are of duration $t_{s}$.
In the following proposition we characterize the $\overline{A}^{\pi_{t_{s}}}$ and $\overline{P}^{\pi_{t_{s}}}$ with a packet error probability of $\epsilon$ for AT model.

\begin{proposition}
For AT packet generation model with an age-threshold of $h_{a} \geq t_{s}$ and FTT policy with packet transmission duration $t_{s}$ and packet error probability $\epsilon$ we have that
\begin{eqnarray*}
 \overline{A}^{\pi_{t_{s}}} & = & t_{s} + \frac{(h_{a}-t_{s})^{2}(1 - \epsilon)^{2} + t_{s}^{2}(1 + \epsilon) + 2(h_{a} - t_{s})t_{s}(1 - \epsilon)}{2(1 - \epsilon)((h_{a}-t_{s})(1 - \epsilon) + t_{s})} - \frac{1}{2}, \\
 \overline{P}^{\pi_{t_{s}}} & = & \frac{P(t_{s})t_{s}}{(h_{a} - t_{s})(1 - \epsilon) + t_{s}}
\end{eqnarray*}
The AAoI and average power for any $h_{a} < t_{s}$ is the same as that for $h_{a} = t_{s}$.
\label{prop:ATtradeoff}
\end{proposition}
The proof is discussed in Appendix \ref{appendix:proofATtradeoff}.

\subsection{Results and discussion}
We compare the AAoI-power tradeoff for FTT policies for the above packet generation models using the analytical 
characterizations from Propositions \ref{proposition_packet_loss}, \ref{prop:Pmodel_tradeoff}, and \ref{prop:ATtradeoff} 
in Figure \ref{fig:tradeoff_packetgen}.
We note that the parameter $\lambda$ is $0.01$ in Figure \ref{fig:tradeoff_packetgen} for NP and P models.
We also plot the Pareto achievable tradeoff for P and NP models (denoted as P-OPT and NP-OPT respectively) obtained by 
minimizing a linear combination of $\overline{A}^{\pi_{t_{s}}}$ and $\overline{P}^{\pi_{t_{s}}}$ over $t_{s}$ and 
$\lambda$.
The minimization is done using Differential Evolution (DE) (which leads to a local minima).
Other parameters are chosen as in Section \ref{low_reliability_results}.

We observe that the AT model achieves the best tradeoff while the P model for a fixed $\lambda$ has the worst 
performance.
We observe from Figure \ref{fig:tradeoff_packetgen_highpower} that NP-OPT and AT has similar tradeoff performance for 
higher average power values.

For lower values of $\lambda$ we expect that the tradeoff performance of P model would be similar to that of NP model 
since the pre-emptions would be rare.
From Propositions \ref{proposition_packet_loss} and \ref{prop:Pmodel_tradeoff} we observe that as $\lambda 
\downarrow 0$, for both NP and P models, $\overline{P}^{\pi_{t_{s}}}$ is $\mathcal{O}(\lambda)$ and 
$\overline{A}^{\pi_{t_{s}}}$ grows as $\frac{1}{\lambda}$ with a coefficient of $\frac{1}{1 - \epsilon}$.
For the AT model, $\overline{P}^{\pi_{t_{s}}} \downarrow 0$ as $h_{a} \uparrow \infty$ in which case 
$\overline{A}^{\pi_{t_{s}}}$ is $\overline{O}(h_{a})$ with the same coefficient $1/(1 - \epsilon)$.

For higher values of $\lambda$ we observe (not reported here) that the P model has $\overline{A}^{\pi_{t_{s}}}$ which 
is magnitudes higher than that of AT and NP models since transmissions are often pre-empted.
The AT model is observed to achieve a better tradeoff performance even in this case.
Further analysis of the tradeoff under the AT model is part of future work. 
\begin{figure}
    \centering
    \begin{subfigure}{0.49\linewidth}
    \centering
    \includegraphics[width = 1\linewidth]{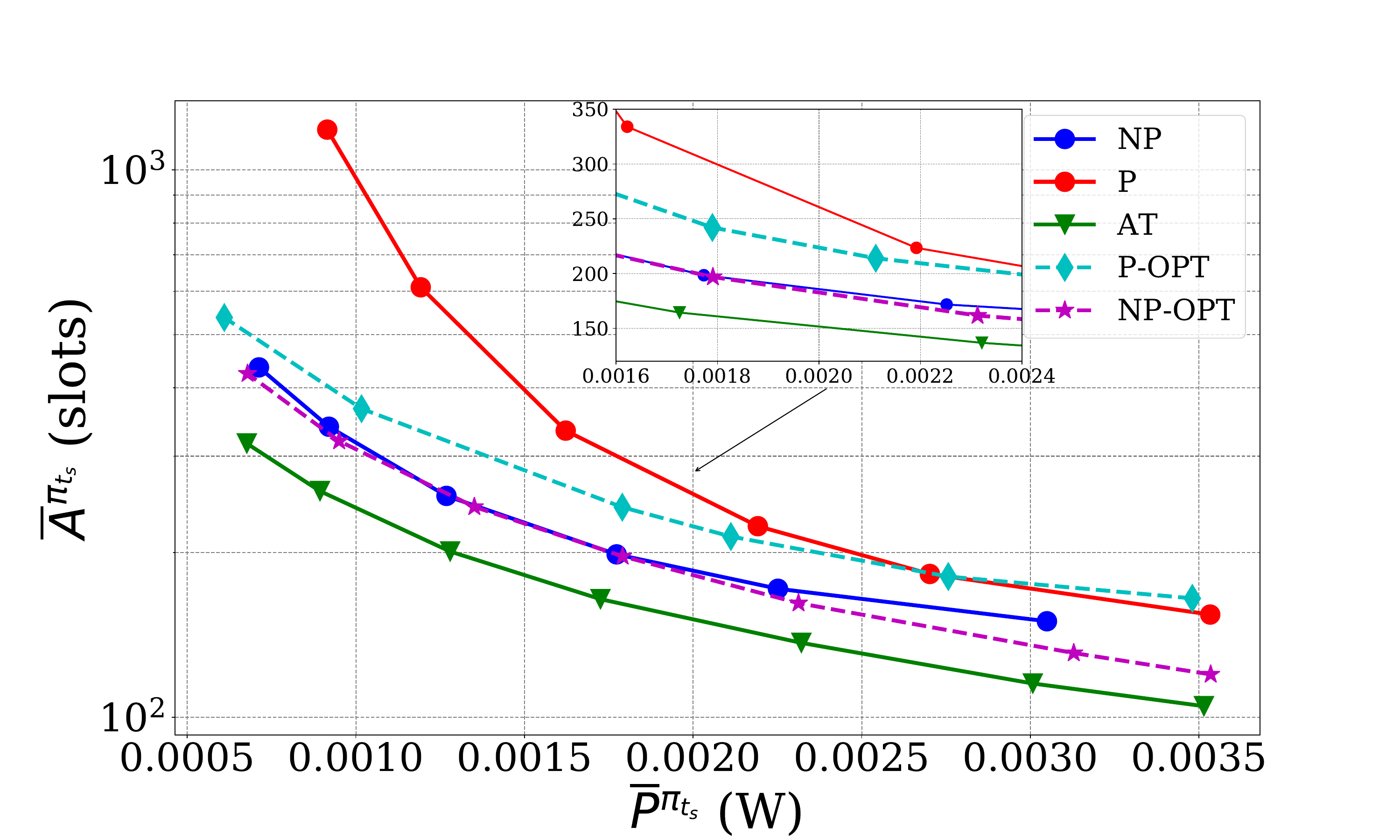}
    \caption{$\epsilon = 0.01$}
    \label{fig:tradeoff_packetgen_eps001}
    \end{subfigure}
    \begin{subfigure}{0.49\linewidth}
    \includegraphics[width = 1\linewidth]{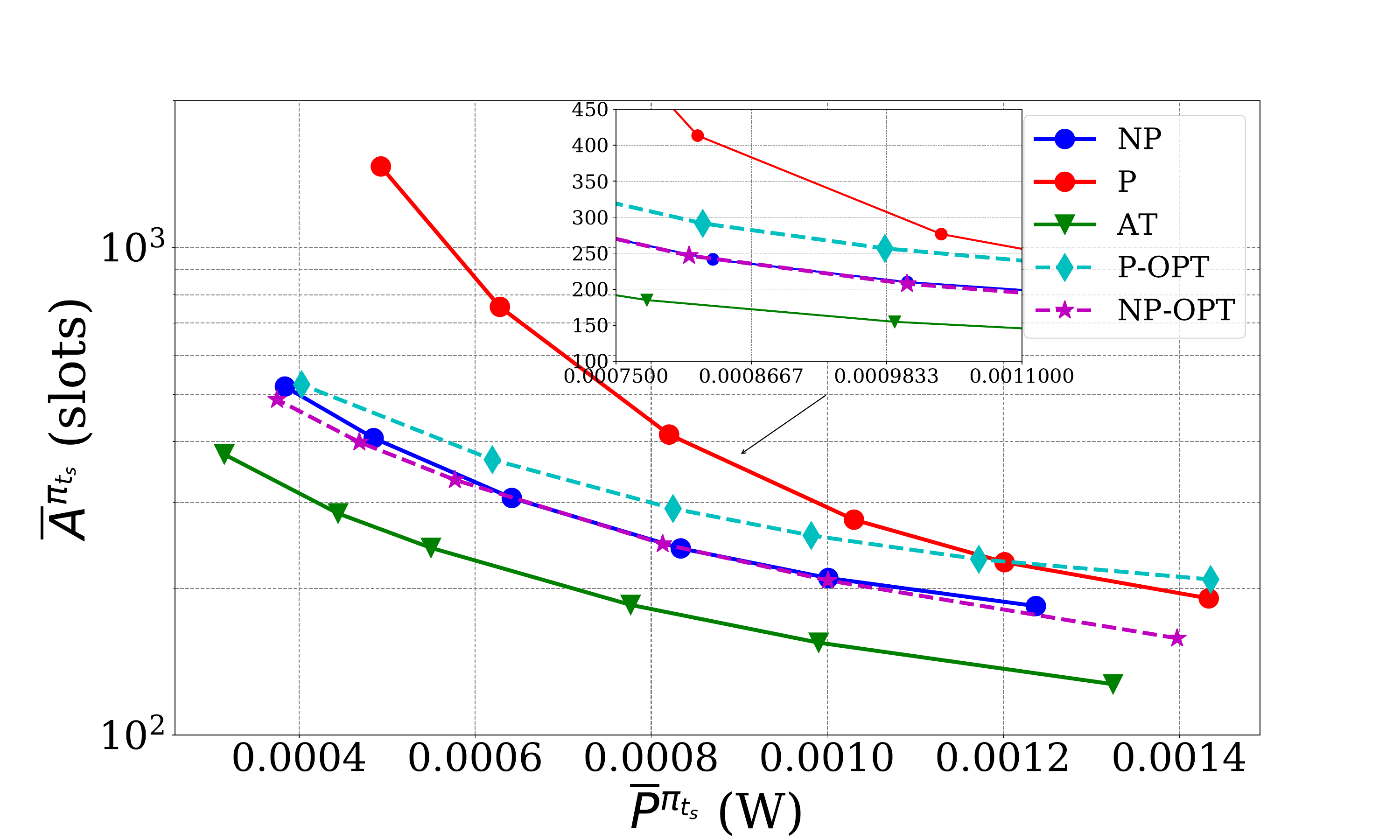}
    \caption{$\epsilon = 0.2$}
    \label{fig:tradeoff_packetgen_eps02}
    \end{subfigure}
    \caption{Illustration of the tradeoff between $\overline{A}^{\pi_{t_{s}}}$ and 
$\overline{P}^{\pi_{t_{s}}}$ for FTT policies with the packet generation models: NP, P, and AT. The parameter $\lambda$ 
for NP and P is chosen to be $0.01$. The parameter $h_{a}$ and $t_{s}$ for AT has been chosen to obtain the Pareto 
tradeoff. NP-OPT and P-OPT are the Pareto tradeoffs for NP and P respectively obtained by using Differential Evolution 
to choose $\lambda$ and $t_{s}$.}
    \label{fig:tradeoff_packetgen}
\end{figure}

\begin{figure}
    \centering
    \includegraphics[width = 0.5\linewidth]{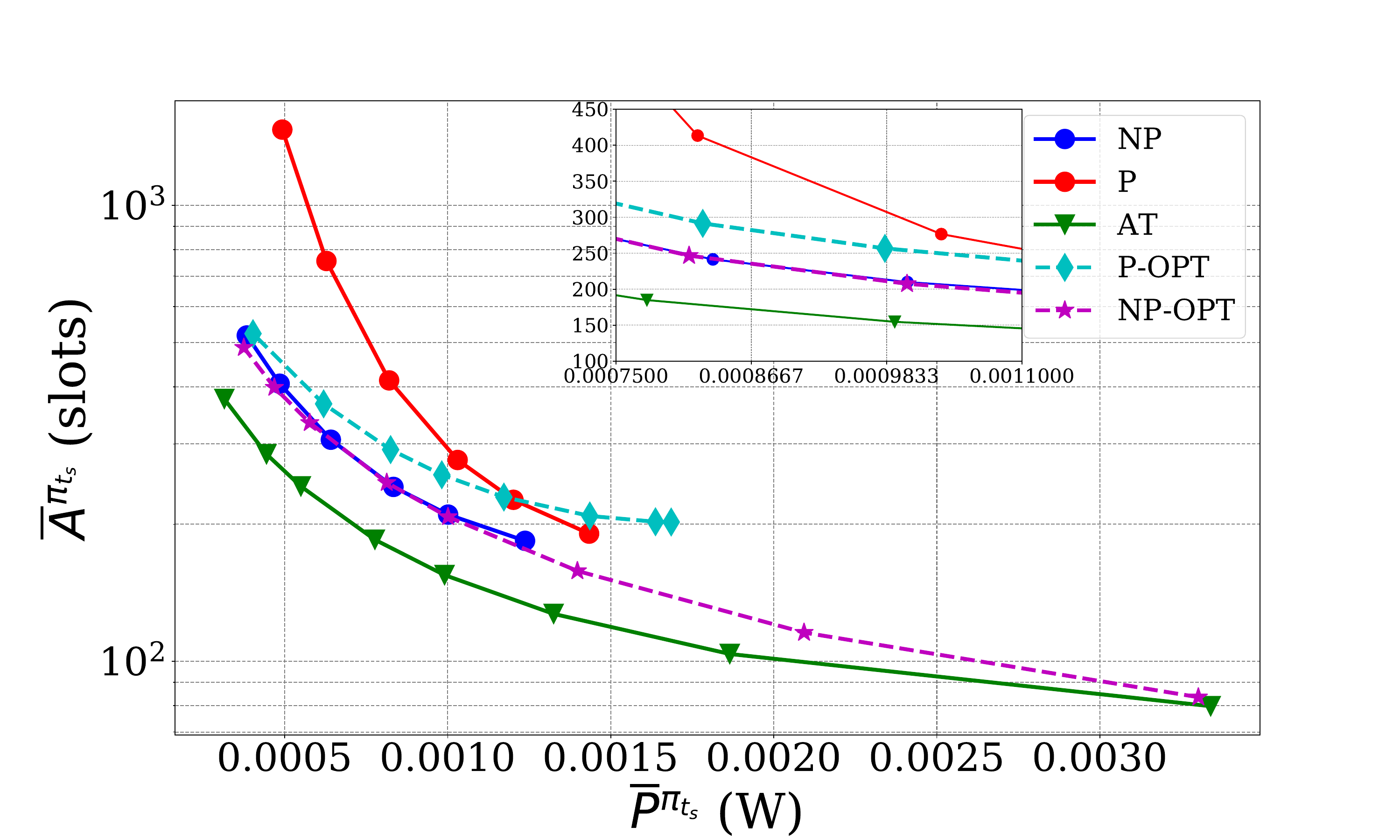}
    \caption{Illustration of the tradeoff between $\overline{A}^{\pi_{t_{s}}}$ and 
$\overline{P}^{\pi_{t_{s}}}$ for FTT policies with the packet generation models: NP, P, and AT including larger values 
for power. The parameter $\lambda$ for NP and P is chosen to be $0.01$. The parameter $h_{a}$ and $t_{s}$ for AT has 
been chosen to obtain the Pareto tradeoff. NP-OPT and P-OPT are the Pareto tradeoffs for NP and P respectively obtained 
by using Differential Evolution to choose $\lambda$ and $t_{s}$.}
    \label{fig:tradeoff_packetgen_highpower}.
\end{figure}

\section{Conclusions and Future Work}
In this paper, we considered the tradeoff of AAoI and transmit power for a point-to-point link where the scheduler at the transmitter has the capability to dynamically adapt the block-length.
We showed that in the regime of high reliability ($\epsilon \approx 0$) and high $\lambda$, non-adaptive fixed transmission time/duration (FTT) policies are close to optimal.
For a model where all transmissions are received without error and $\lambda = 1$, FTT policies were proved to be optimal.
Furthermore, FTT policies (with time-sharing) or state-independent randomized policies were shown to be order optimal for highly reliable systems.
A characterization of the tradeoff can be obtained using an analytical upper bound (from the AAoI-power characterization for FTT policies) as well as numerical, analytical, and  asymptotic lower bounds.
The asymptotic lower bound is used to obtain the above order-optimality result.
Such bounds would enable performance evaluation of other state-dependent heuristic policies which could achieve a better tradeoff performance compared with FTT.
For example, for low packet generation rates, threshold policies were observed to have a better tradeoff performance.
An approximate analytical characterization of the threshold policy has also been presented for high reliability.
We also note that the analytical characterizations of the tradeoff for FTT and threshold policies help to guide the choice of policy parameters.
We considered a wireless point-to-point link with block fading, where the effect of channel coherence times on the achieved tradeoff for FTT policies was studied.
Shorter channel coherence times which enable channel diversity were observed to have a better tradeoff.
The AoI metric can also be optimized by scheduling packet transmissions. 
To understand the effect of packet transmission scheduling on the tradeoff, we analytically characterized the tradeoff performance for three packet generation models with fixed transmission times and observed that a state dependent packet generation scheme (age-threshold scheme) had the best tradeoff.
Understanding the tradeoff for a joint packet-generation and packet-transmission scheduling policy (with block-length adaptation) is planned for future work.
We also note that there might be regimes where state dependent policies such as threshold policies have better performance compared to FTT policies.
identifying such regimes is part of future work.
In this paper, we had considered a point-to-point link, however extending the tradeoff analysis to a general radio resource block allocation policy has scope for future work.

\appendices
\section{Proof of Proposition~\ref{proposition_packet_loss}}
\label{proof_proposition_packet_loss}
    We obtain the average AoI and average power using the renewal reward theorem (RRT)\cite{kumar2008wireless}. For applying RRT, we first identify a renewal process in the evolution of $A[t]$ under an FTT policy with parameter $t_{s}$. We define a renewal epoch as the slot in which the age $A[t]$ drops due to a packet's reception(refer Figure \ref{fig:packet_loss}).
    To be precise, this is the start of the next slot after a packet finishes transmission.
    This is motivated by the fact that the age drops to $t_{s}$ at every successful reception as we consider FTT policy.
    Let $X$ be the random variable denoting the number of transmission failures in one renewal cycle. Then $X$ is geometrically distributed with parameter $(1-\varepsilon)$. 
    Therefore, $P_{X}(x) = \varepsilon^{x}(1-\varepsilon)$.
    \begin{figure}[!htb]
        \centering
        \includegraphics[width = 0.6\linewidth]{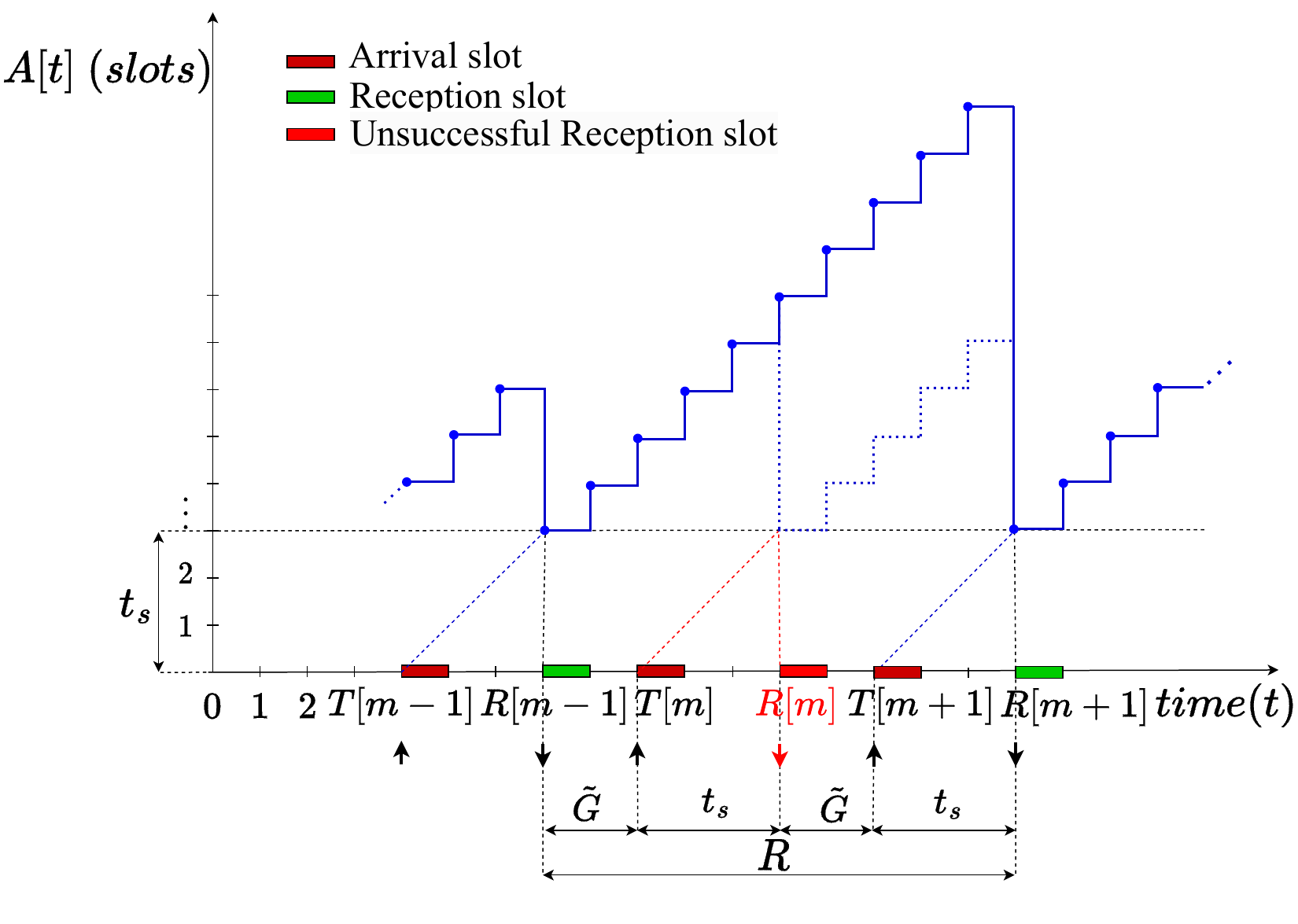}
        \caption{Illustration of the evolution of AoI $A[t]$ under FTT policy with parameter $t_{s}$. The AoI $A[t]$ drops when a packet is successfully received. The effect of an error leading to an unsuccessful reception of a packet is also shown.}
        \label{fig:packet_loss}
    \end{figure}
    In Figure ~\ref{fig:packet_loss}, we illustrate the evolution of $A[t]$ when the FTT with parameter $t_{s}$ policy is adopted in the non-preemptive scheme for the model with packet losses. The $m$\textsuperscript{th} packet is lost, and the transmitter initiates sampling of a new packet at time $R[m]$, which is generated at $T[m+1]$ after a random delay $\tilde{G}\sim\text{geometric}(\lambda)$.
    We have that
    \begin{equation*}
        \Exp{[X]} = \frac{\varepsilon}{1-\varepsilon}, \quad \operatorname{Var}(X) = \frac{\varepsilon}{(1-\varepsilon)^2}
    \end{equation*}
    Referring Figure ~\ref{fig:packet_loss},  the renewal cycle is
    \begin{equation*}
        R = \sum_{i = 1}^{X+1}t_{s}+\tilde{G} = (X+1)(t_{s}+\tilde{G})
    \end{equation*}
    Therefore,
    \begin{equation*}
        \Exp{[R]} = \frac{1}{1-\varepsilon}\Brap{\frac{1-\lambda}{\lambda} + t_{s}}.
    \end{equation*}
    We compute $\Exp{[R^2]}$ as $\operatorname{Var}(R) + (\mathbb{E}[R])^{2}$. We have that
    \begin{align*}
    \operatorname{Var}(R) =&\operatorname{Var}\Brap{\sum_{i = 1}^{X+1}t_{s}+\tilde{G}} \\
    =&~\Exp{\Bras{\operatorname{Var}\Brap{\sum_{i = 1}^{X+1}t_{s}+\tilde{G}\mid X=x}}} \\ &+\operatorname{Var}\Brap{\Exp\Bras{\sum_{i = 1}^{X+1}t_{s}+\tilde{G} \mid X=x}} \\
    =& \frac{1-\lambda}{(1-\varepsilon)\lambda^{2}}+\frac{\varepsilon}{(1-\varepsilon)^{2}}\Bras{\frac{1-\lambda}{\lambda}+t_{s}}^{2} 
    \end{align*}
    Finally, we have that
    \begin{align*}
    \Exp R^{2} =~& \frac{1-\lambda}{(1-\varepsilon)\lambda^{2}}+\frac{1+\varepsilon}{(1-\varepsilon)^{2}}\Bras{\frac{1-\lambda}{\lambda}+t_{s}}^2.
    \end{align*}
    Similarly, we derive average power using RRT. The power consumed for transmitting is $P(t_{s})$. The expected duration over which this is consumed is $\frac{t_{s}}{1-\varepsilon}$.
    Therefore, by applying RRT, we have that the average power is 
    \begin{equation*}
    \overline{P}^{\pi_{t_{s}}}_{L} = \frac{P(t_{s})t_{s}\lambda}{1-\lambda + \lambda t_{s}}.
    \end{equation*}

\section{Proof of Proposition \ref{prop:lowpower_lb}}
\label{app:proof:proposition_orderoptimality}

\begin{proof}
We note that $P_{min}$ is obtained when the maximum transmission duration $\tau_{max}$ is used for all the 
transmissions. i.e., under the FTT policy with $t_s = \tau_{max}$. 
Therefore the minimum average power $P_{min} = \frac{P(\tau_{max})\tau_{max}} {\tau_{max}+\Exp{\tilde{G}}}$.
Suppose the average power constraint $p_{c}$ is $\delta$ more than the minimum achievable average power $P_{min}$.
In this proof, we consider a specific $k$ so that the subscript $k$ is dropped in the notation for brevity (e.g., 
$\delta_{k}$ is denoted as $\delta$).
For any stationary policy $\pi$ such that $\overline{P}^{\pi} \leq p_{c}$ we have (as in the proof of Proposition 
\ref{prop:numerical_lb})
\begin{align*}
\frac{\sum_{\tau}\sum_{\tau^{\prime}}p(\tau,\tau^{\prime})P(\tau)\tau}{\sum_{\tau}\sum_{\tau^{\prime}}p(\tau,\tau^{
\prime})\tau + \Exp{\tilde{G}}} \le \frac{P(\tau_{max})\tau_{max}}{\tau_{max} + \Exp{\tilde{G}}} + \delta.
\end{align*}
That is,
\begin{align*}\sum_{\tau}\sum_{\tau^{\prime}}p(\tau,\tau^{\prime})\left[P(\tau)\tau - P(\tau_{max})\tau_{max}\right] 
\le 
\delta(\tau_{max} + \Exp{\tilde{G}}).
\end{align*}
 As $P(\tau)\tau$ is monotonically decreasing, we obtain the following upper bound on the joint distribution $p(\tau, 
\tau^{\prime})$, $\forall \tau \ne \tau_{max}$.
 \begin{equation}p(\tau, \tau^{\prime}) \le \frac{\delta(\tau_{max}+\Exp{\tilde{G}})}{P(\tau)\tau -
P(\tau_{max})\tau_{max}} \label{eq:LB_joint}\end{equation}
From the constraint on power, we also have that
\begin{align*}
\frac{\sum_{\tau}\pi_{\tau}P(\tau)\tau}{\tau_{max} + \Exp{\tilde{G}}} &\le \frac{P(\tau_{max})\tau_{max}}{\tau_{max} + 
\Exp{\tilde{G}}} + \delta.
\end{align*}
This implies that the stationary distribution $\pi_{\tau}$ of using a transmission duration $\tau$ satisfies
\begin{align*}
\sum_{\tau}\pi_{\tau}[P(\tau)\tau - P(\tau_{max})\tau_{max}] &\le \delta(\tau_{max} + \Exp{\tilde{G}}).
\end{align*}
Therefore $\forall \tau \ne \tau_{max}$, we have the following upper bound on $\pi_\tau$.
\begin{equation}\pi_{\tau} \le \LBterm{\tau} \label{eq:LB_marginal}\end{equation}

 From the proof of Proposition \ref{prop:numerical_lb}, for an error-free system we have that AAoI is:
\begin{equation}\frac{\sum\limits_{\tau}\sum\limits_{\tau^{\prime}}p(\tau,\tau^{\prime})\Bras{\tau\tau^{\prime} + 
\tau'\Exp{\tilde{G}} + \frac{\tau^{\prime}(\tau^{\prime}-1)}{2} + \tau^{\prime}\Exp{\tilde{G}} + 
\Exp{\Bras{\frac{\tilde{G}(\tilde{G}-1)}{2}}}}}{\tau_{max} + \Exp{\tilde{G}}}.
 \label{eq:objective}
 \end{equation}
 For an error free system, we have that the age at an arrival epoch is $\tau + \tilde{G}$, where $\tau$ represents the last transmission duration duration.
Now we consider each term (inside the bracket multiplied by $p(\tau, \tau')$) of the numerator and lower bound each 
using the inequalities \eqref{eq:LB_joint} and
\eqref{eq:LB_marginal}. 
The expansions of each term of \eqref{eq:objective} after bounding are as follows.\\
\noindent
\textit{Term-1}: $\sum_{\tau}\sum_{\tau^{\prime}}p(\tau,\tau^{\prime})\tau\tau^{\prime}$ is bounded below by
\begin{align*}
&\sum_{\tau}\sum_{\tau^{\prime}\ne \tau_{max}} \LBterm{\tau} \tau\tau^{\prime} + \\ 
&\tau_{max}\Bras{\tau_{max}-\sum_{\tau\ne \tau_{max}}\sum_{\tau^{\prime}=\tau_{max}}(\tau_{max}-\tau)\LBterm{\tau}}
\end{align*}
\textit{Term-2}:
$\sum_{\tau}\sum_{\tau^{\prime}}p(\tau,\tau^{\prime})\tau'\Exp{\tilde{G}}$ which is 
$\Exp{\tilde{G}}\sum_{\tau}\pi_{\tau}\tau$ is bounded below by
$$
\Exp{\tilde{G}}\Bras{\tau_{max}-\sum_{\tau\ne\tau_{max}}(\tau_{max}-\tau)\LBterm{\tau}}
$$
\textit{Term-3}:
$\sum_{\tau}\sum_{\tau^{\prime}}p(\tau,\tau^{\prime})\frac{{\tau^{\prime}}^{2}}{2}$ which is 
$\frac{1}{2}\sum_{\tau}\pi_{\tau}\tau^{2}$  is bounded below by
$$\frac{1}{2}\Bras{\tau_{max}^{2}-\sum_{\tau^{\prime}\ne\tau_{max}}(\tau_{max}^{2} - {\tau^{\prime}}^{2}) 
\LBterm{\tau^{\prime}}}$$
\textit{Term-4}:
$-\sum_{\tau}\sum_{\tau^{\prime}}p(\tau,\tau^{\prime})\frac{\tau^{\prime}}{2}$ which is 
$-\frac{1}{2}\sum_{\tau}\pi_{\tau}\tau$ is bounded below by
$$-\frac{1}{2}\Bras{\tau_{max}-\sum_{\tau^{\prime}\ne\tau_{max}}(\tau_{max}-\tau^{\prime})\LBterm{\tau^{\prime}}}$$
\textit{Term-5}:
$\sum_{\tau}\sum_{\tau^{\prime}}p(\tau,\tau^{\prime})\tau^{\prime}\Exp{\tilde{G}}$ (similar to Term-2) is bounded below 
by
$$
\Exp{\tilde{G}}\Bras{\tau_{max}-\sum_{\tau^{\prime}\ne\tau_{max}}(\tau_{max}-\tau^{\prime})\LBterm{\tau^{\prime}}}
$$
Summing together all the above lower bounds and dividing by $\tau_{max} + \Exp{\tilde{G}}$, we obtain a lower bound on 
$\overline{A}^{\pi}$.

We also observe that the sum of the first term in each of the lower bound expressions divided by $\tau_{max} + 
\Exp{\tilde{G}}$ is $\overline{A}^{\pi_{\tau_{max}}}$.
The other terms are all $\mathcal{O}(\delta)$.
Therefore, we obtain that $\overline{A}^{\pi_{\tau_{max}}} - \overline{A}^{\pi}$ is $\mathcal{O}(\delta)$.
\end{proof}

\section{Bounds on stationary probability for a high-power regime}
\label{sec:app:bounds_highpower}
Consider any policy with average power more than $P_{max} - \delta$.
If $\pi_{\tau}$ represents the stationary probability of using transmission duration $\tau$, then
\begin{eqnarray*}
  \frac{P(\tau_{min})\tau_{min}}{\tau_{min} + \frac{1- \lambda}{\lambda}} - \frac{\sum_{\tau} \pi_{\tau} 
P(\tau)\tau}{\sum_{\tau} \pi_{\tau} \tau + \frac{1 - \lambda}{\lambda}} \leq \delta.
 \end{eqnarray*}
 That is,
 \begin{eqnarray*}
  \frac{\sum_{\tau} \pi_{\tau} P(\tau_{min})\tau_{min}}{\tau_{min} + \frac{1- \lambda}{\lambda}} - \frac{\sum_{\tau} 
\pi_{\tau} 
P(\tau)\tau}{\sum_{\tau} \pi_{\tau} \tau_{min} + \frac{1 - \lambda}{\lambda}} \leq \delta.
 \end{eqnarray*}
 Which implies that for $\tau \neq \tau_{min}$
 \begin{eqnarray*}
  \pi_{\tau} \leq \frac{\delta \brap{\tau_{min} + \frac{1- \lambda}{\lambda}}}{P(\tau_{min})\tau_{min} - P(\tau)\tau}.
 \end{eqnarray*}
 Similar to the proof of Proposition \ref{prop:lowpower_lb} such bounds can also be derived for the joint probability $p(\tau, \tau')$, which leads to a similar order-optimality result in the high-power regime where $p_{c} \uparrow P_{max}$.

\section{AAoI-Power Tradeoff for error-free system} 
\label{sec:errorfree_system}
In this section, we consider an error-free variation (denoted as ERRFREE) of the system model discussed in Section 
\ref{section:System_Model}.
The ERRFREE assumes that every transmission is error free independent of the choice of the transmission 
time $\tau_{m}$.
However, the power-transmission duration relationship (i.e, the $\tau$-$P(\tau)$ relationship) is assumed to be the same as 
in the original system model.
In this section, the original system is denoted as WITHERR.

We are motivated to study ERRFREE because of its analytical tractability.
The analysis of ERRFREE enables us to obtain useful analytical approximations for the tradeoff performance of 
policies such as threshold policies for WITHERR systems.
These approximations are used in the design (choice of parameters) of such policies.

First of all, we note that a SMDP approach (similar to that in Section~\ref{section:SMDP_packet_loss}) can be used to 
numerically characterize the tradeoff for ERRFREE.
The state and action spaces for ERRFREE are the same as in Section~\ref{section:SMDP_packet_loss}.
The decision times of the SMDP coincide with the generation times of the packets. 
The expected time between two consecutive decision epochs is:
\begin{equation*}
    \tilde{\tau}(A_{m}, \tau(A_{m}) = \tau) = \tau + \Exp{\tilde{G}}.
\end{equation*}
Since transmissions are error-free, $A_{m+1} = \tau(A_{m}) + 
\tilde{G}$ (refer Figure \ref{fig:age_transition_NP}).
\begin{figure}[!htb]
\centering
\includegraphics[width = 0.6\linewidth]{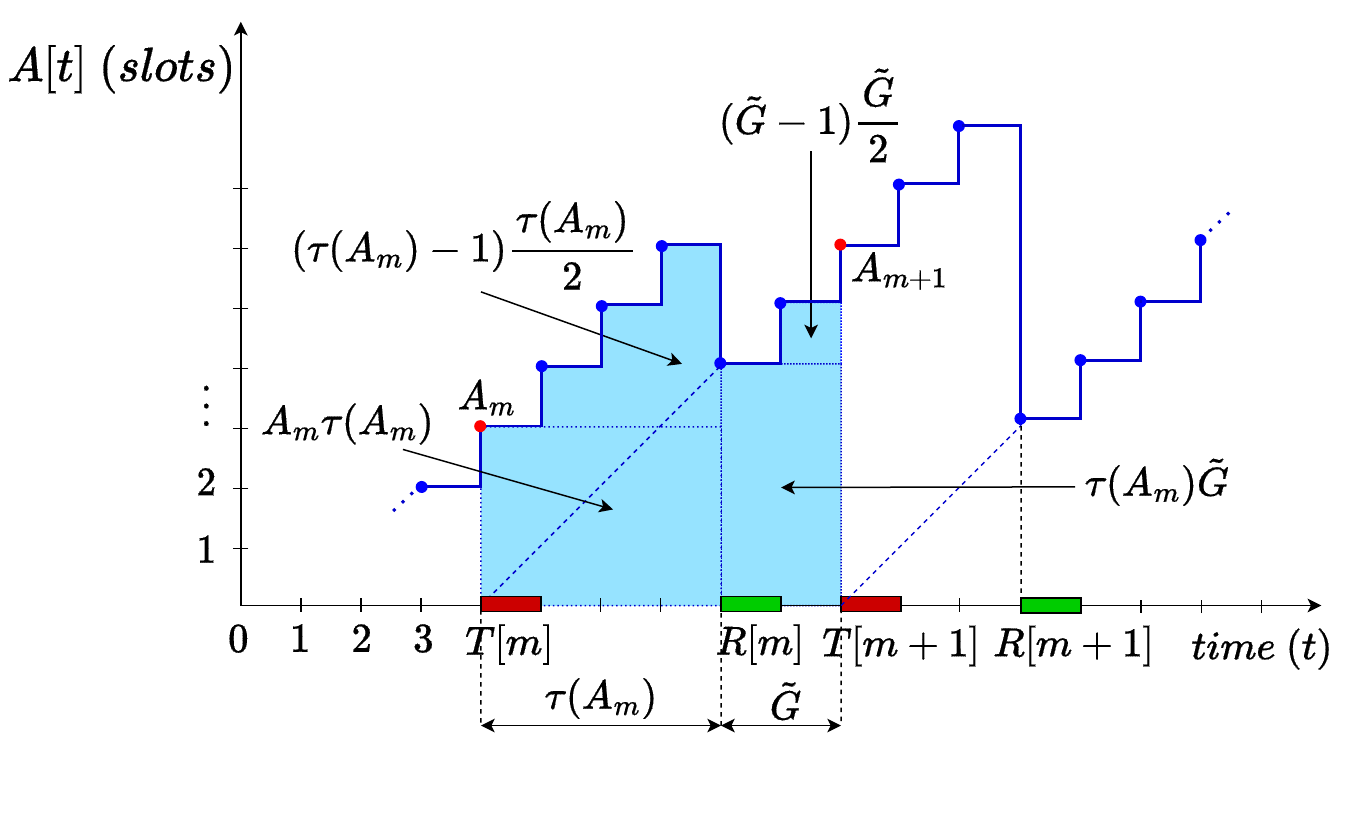}
\caption{Illustration of the evolution of AoI $A[t]$. The transition from $A_{m}$ to $A_{m+1}$ along with the terms contributing to the cumulative age cost is also shown.}
\label{fig:age_transition_NP}
\end{figure}
Therefore the transition probability 
\begin{equation*}
    \operatorname{Pr}(A_{m+1} = a^{\prime} \mid \tau(A_{m}) = \tau) = \begin{cases}
    \lambda (1 - \lambda)^{a^{\prime}-\tau} \text{ for } a^{\prime} \ge \tau, \\
    0 \quad \text{otherwise}.
    \end{cases} 
\end{equation*}
For ERRFREE, we define the single-stage cost $c(a, \tau)$ as
\begin{align*}
      c(a, \tau) & = a\tau + (\tau-1)\frac{\tau}{2} + \tau \frac{1-\lambda}{\lambda} + 
\Brap{\frac{1-\lambda}{\lambda}}^{2} + \beta P(\tau) \tau
\end{align*}
We note that a numerical procedure such as value iteration \cite{TijmsH.C2003Afci} can be used to solve a truncated 
version of the SMDP (where the state or age values are limited to a maximum value $a_{max}$).

For ERRFREE, the AAoI and average power for FTT policies can be obtained from Proposition \ref{proposition_packet_loss} 
with $\epsilon = 0$.
We also have a characterization of the AAoI and average power for threshold policies for ERRFREE which is discussed 
next.

Consider the evolution of $A[t]$ for a threshold policy with parameters $h$, $\tau_{a}$, and $\tau_{b}$ in ERRFREE.
We note that at any decision epoch, the transmission duration chosen is either $\tau_{a}$ or $\tau_{b}$.
Then, at the end of that transmission, the age $A[t]$ is therefore either $\tau_{a}$ or $\tau_{b}$ respectively (since 
transmissions are error free).
Consider the slots which just after the end of a transmission.
The age at these slots are denoted as $A_{e,m}$, where $m \in \sZ$ indexes the transmission end-time slots.
Then, $A_{e,m} \in \brac{\tau_{a}, \tau_{b}}, \forall m$.
We note that the evolution from $A_{e,m}$ to $A_{e,m+1}$ is independent of the past evolution of the age given 
$A_{e,m}$.
Given $A_{e,m}$ the next transmission duration is chosen based on the age value $A_{m+1}$ at the next decision epoch, which 
is $A_{e,m} + \tilde{G}$, where $\tilde{G} \in \brac{0,1,2,\dots}$ is Geometric$(\lambda)$ and sampled independently of 
anything else.
Thus, $A_{e,m}$ is an EMC embedded in the evolution of $A[t]$.

The transition probability of the EMC ($A_{e,m}$) depends on the relationship between $h$ and the transmission 
times $\tau_{a}$ and $\tau_{b}$.
We use the notation indicated in the transition diagram in Figure ~\ref{fig:MRP_STD}.
The stationary probabilities of the states $\tau_{a}$ and $\tau_{b}$ are $\pi_{\tau_{a}} = 
\frac{\alpha}{\alpha + \beta}$ and $\pi_{\tau_{b}} = \frac{\beta}{\alpha + \beta}$ respectively.

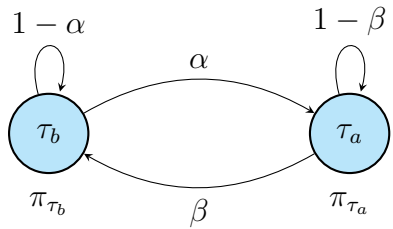
\begin{figure}[!h]
\centering
\begin{tikzpicture}[node distance = 4cm]
\tikzset{->, >=stealth, every state/.style={thick, fill=cyan!25}}
\node[state] (tau_b) {$\tau_{b}$};
\node[state, right of=tau_b] (tau_a) {$\tau_{a}$};
\draw (tau_b) edge[loop above] node{$1-\alpha$} (tau_a)
(tau_a) edge[loop above] node{$1-\beta$} (tau_b)
(tau_b) edge[bend left, above] node{$\alpha$} (tau_a)
(tau_a) edge[bend left, below] node{$\beta$} (tau_b);
\node[below=0.1cm of tau_b] (tau_b) {$\pi_{\tau_{b}}$};
\node[below=0.1cm of tau_a] (tau_a) {$\pi_{\tau_{a}}$};
\end{tikzpicture}
\caption{Transition probability diagram of the two-state EMC $(A_{e,m})$ for the threshold policy.}
\label{fig:MRP_STD}
\end{figure}

\textit{Case 1: $h \ge \tau_{a}$}:
The transition probabilities are as follows:
\begin{align*}
\alpha &= \operatorname{Pr}(\tilde{G} \le h - \tau_{b}) = 1 - (1-\lambda)^{h-\tau_{b}+1}, \\
\beta  &= \operatorname{Pr}(\tilde{G} > h - \tau_{a}) = (1-\lambda)^{h-\tau_{a}+1}.
\end{align*}
The state transitions from $\tau_{b}$ to $\tau_{a}$ if the age at the next decision epoch has not increased enough to 
cross the threshold $h$, similarly the state transitions from $\tau_{a}$ to $\tau_{b}$ if the age at the next decision 
epoch crosses the threshold $h$.
The age at the next decision epoch is $\tilde{G}$ more than the current state; this is used to obtain the transition 
probabilities above. 

\textit{Case 2: $h < \tau_{b}$}: 
In this case, the age at a decision epoch would always be greater than $h$, since $A_{e,m} \geq \tau_{b}$.
Since the threshold policy chooses $\tau_{b}$ if the age is greater than $h$, $\alpha = 0$ and $\pi_{\tau_{b}} = 1$.

\textit{Case 3: $\tau_{b} \le h < \tau_{a}$}: 
Consider the state being at $\tau_{a}$, then since $h < \tau_{a}$, at the next decision epoch, we would choose 
$\tau_{b}$ with probability $1$. Hence, $\beta = 1$.
If the state is $\tau_{b}$, then the transition probability $\alpha$ is $1 - (1-\lambda)^{h-\tau_{b}+1}$ as in Case 1.

In each case, the stationary probabilities can be computed using $\alpha$ and $\beta$.
In the following proposition, we use MRRT for $A_{e,m}$ to characterize the AAoI and average power for a threshold 
policy.

\begin{proposition}
\label{prop:threshold}
For a threshold policy with parameters $h$, $\tau_{a}$, and $\tau_{b}$, the AAoI is
\begin{equation}\label{eq:aaoi_thresh}
    \bar{A}^{\pi_{h}} = \frac{\pi_{\tau_{b}}c^{a}_{\tau_{b}} + 
\pi_{\tau_{a}}c^{a}_{\tau_{a}}}{\pi_{\tau_{b}}T_{\tau_{b}} + \pi_{\tau_{a}}T_{\tau_{a}}},
\end{equation}
and the average transmit power is
\begin{equation}\label{eq:avg_p_thresh}
    \bar{P}^{\pi_{h}} = \frac{\pi_{\tau_{b}}c^{p}_{\tau_{b}} + 
\pi_{\tau_{a}}c^{p}_{\tau_{a}}}{\pi_{\tau_{b}}T_{\tau_{b}} + \pi_{\tau_{a}}T_{\tau_{a}}}.
\end{equation}
Here $\pi_{\tau_{a}}$ and $\pi_{\tau_{b}}$ are the stationary probabilities of the states of the EMC $(A_{e,m})$.
The terms $c^{a}_{\tau_{a}}$, $c^{a}_{\tau_{b}}$ are the expected cumulative age costs in a renewal cycle, conditioned 
on the EMC state being $\tau_{a}$ and $\tau_{b}$ respectively. 
Similarly, the terms $c^{p}_{\tau_{a}}$, $c^{p}_{\tau_{b}}$ are the expected cumulative power costs in a renewal cycle, 
conditioned on the EMC state being $\tau_{a}$ and $\tau_{b}$ respectively. 
The terms $T_{\tau_{a}}$ and $T_{\tau_{b}}$ are the renewal cycle durations, again conditioned 
on the EMC state being $\tau_{a}$ and $\tau_{b}$ respectively.
\end{proposition}
\begin{proof}
Assuming the initial state of the MRP to be $\tau_{b}$, the expected cumulative age cost in an interval that begins in 
state $\tau_{b}$, $c^{a}_{\tau_{b}}$ is
\begin{align*}
     & \alpha  \Exp{\biggl[\tau_{b}(\tilde{G}+\tau_{a}) + \frac{(\tilde{G}+\tau_{a})(\tilde{G}+\tau_{a}-1)}{2}\bigg\vert 
\tilde{G} \le h - \tau_{b}\biggr]} + \\ & (1-\alpha) \Exp{\biggl[\tau_{b}(\tilde{G}+\tau_{b}) + 
\frac{(\tilde{G}+\tau_{b})(\tilde{G}+\tau_{b}-1)}{2}\bigg\vert \tilde{G} > h - \tau_{b}\biggr]}.
\end{align*}
The expected cumulative power cost in an interval that begins in state $\tau_{b}$ is
\begin{align*}
    c^{p}_{\tau_{b}} &= \alpha \tau_{a} P(\tau_{a}) + (1-\alpha) \tau_{b} P(\tau_{b})
\end{align*}
and the expected sojourn time until the state transition starting from $\tau_{b}$, $T_{\tau_{b}}$ is
\begin{align*}
    \alpha  \Exp{[\tilde{G}+\tau_{a} \mid \tilde{G} \le h - \tau_{b}]} + (1-\alpha) \Exp{[\tilde{G}+\tau_{b} \mid 
\tilde{G} > h - \tau_{b}]}.
\end{align*}
Similarly, assuming the initial state of the MRP to be $\tau_{a}$, the expected cumulative age cost in an interval that 
begins in state $\tau_{a}$, $c^{a}_{\tau_{a}}$ is 
\begin{align*}
     &\beta \Exp{\biggl[\tau_{a}(\tilde{G}+\tau_{b}) + \frac{(\tilde{G}+\tau_{b})(\tilde{G}+\tau_{b}-1)}{2}\bigg\vert 
\tilde{G} > h - \tau_{a}\biggr]} +\\ & (1 - \beta)  \Exp{\biggl[\tau_{a}(\tilde{G}+\tau_{a}) + 
\frac{(\tilde{G}+\tau_{a})(\tilde{G}+\tau_{a}-1)}{2}\bigg\vert \tilde{G} \le h - \tau_{a}\biggr]}.
\end{align*}
The expected cumulative power cost in an interval that begins in state $\tau_{a}$ is
\begin{align*}
    c^{p}_{\tau_{a}} &= \beta \tau_{b}P(\tau_{b}) + (1-\beta) \tau_{a}P(\tau_{a})
\end{align*}
and the expected sojourn time until the state transition starting from $\tau_{a}$, $T_{\tau_{a}}$ is
\begin{align*}
     \beta \Exp{[\tilde{G}+\tau_{b} \mid \tilde{G} > h - \tau_{a}]} + (1-\beta)\Exp{[\tilde{G}+\tau_{a} \mid \tilde{G} 
\le h - \tau_{a}]}.
\end{align*}
Therefore, leveraging MRRT\cite[Appendix D]{kumar2008wireless}, we obtain expressions for AAoI and average power as 
given in \eqref{eq:aaoi_thresh} and \eqref{eq:avg_p_thresh}, respectively.\\
\end{proof}

\subsection{Numerical \& Simulation Results}\label{section:results}
In this section, we compare the analytical characterization of average AoI and average power obtained for the error 
free system (ERRFREE) with that of the system with errors (WITHERR).
We note that if $\epsilon$ is the error probability for the system with errors, then the $P(\tau)-\tau$ relationship 
for ERRFREE system is also chosen so as to satisfy this error probability requirements (using
\eqref{eq:polyanskiy}).
However, we assume that every transmission is successful in ERRFREE.
For WITHERR, simulations are carried out with errors.

\begin{figure}[h!]
    \centering
    \begin{subfigure}{0.49\linewidth}
    \includegraphics[width = 1\linewidth]{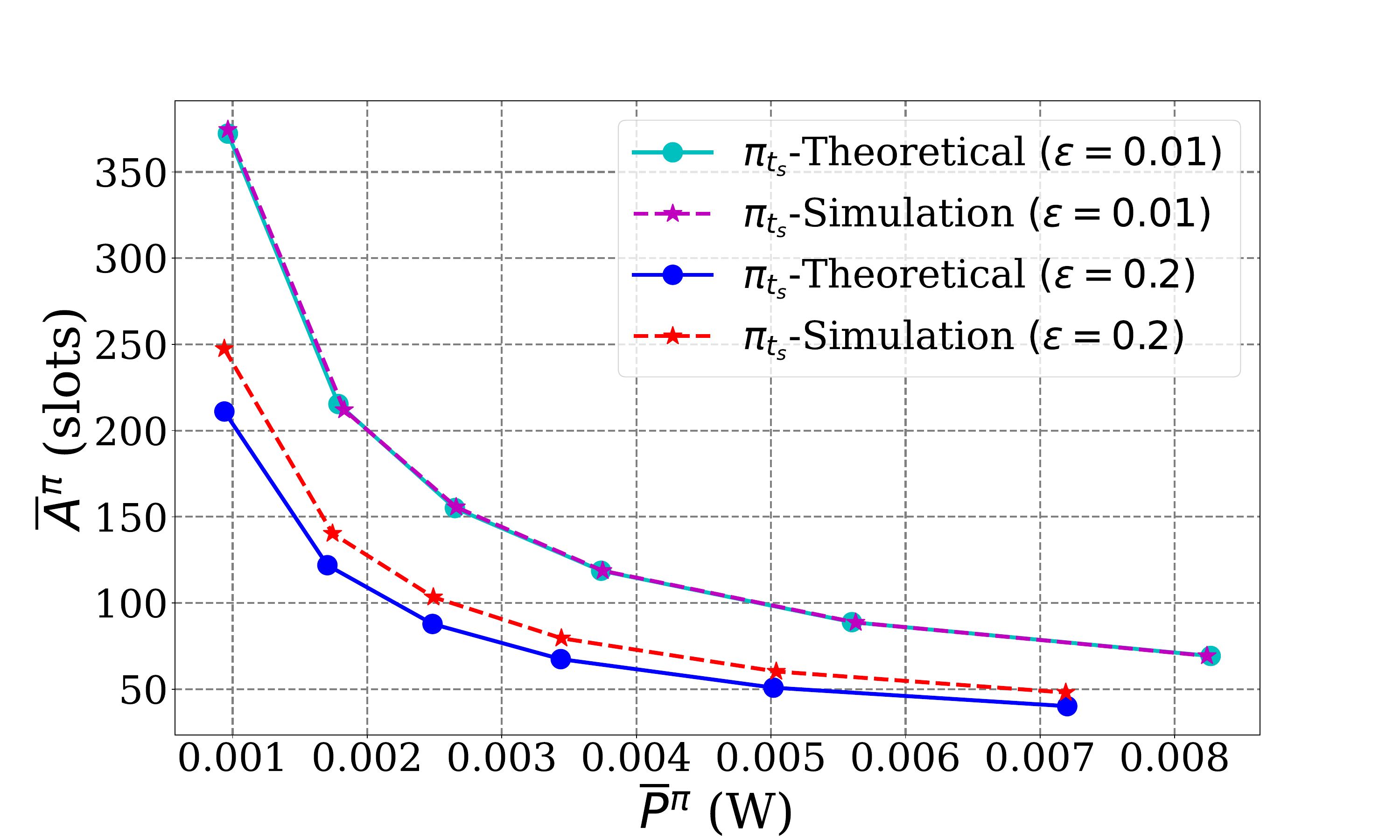}
    \caption{Tradeoff of $\overline{A}^{\pi_{t_{s}}}$ and $\overline{P}^{\pi_{t_{s}}}$}
    \label{fig:fixed_ts_tradeoff_theo_vs_sim_NP}
    \end{subfigure}
    \begin{subfigure}{0.49\linewidth}
    \includegraphics[width = 1\linewidth]{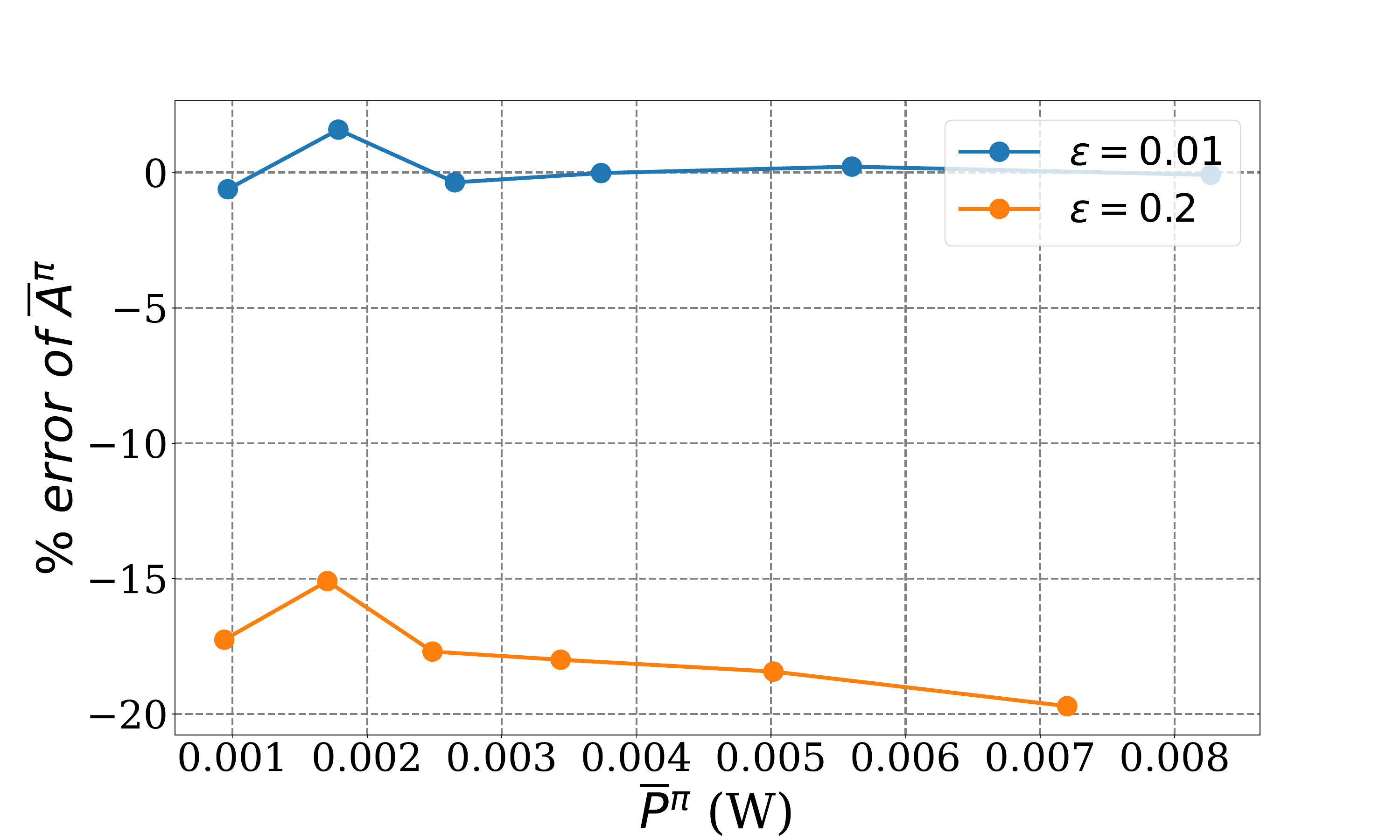}
    \caption{Percentage error in $\overline{A}^{\pi_{t_{s}}}$ calculated with respect to the simulated value}
    \label{fig:fixed_ts_tradeoff_theo_vs_sim_NP_error}
    \end{subfigure}
    \caption{Comparison of AAoI average-power tradeoff for FTT policy obtained from analysis (ERRFREE system) and 
simulation (WITHERR system) for error probabilities $\epsilon \in \brac{0.01,0.2}$. The percentage error between the analytical value and the simulated value of the AAoI calculated with respect to the simulated value at different power values is also shown.}
    \label{fig:comparison_FTT_sim_analytical_error}
\end{figure}

\begin{figure}[h!]
    \centering 
    \begin{subfigure}{0.49\linewidth}
    \includegraphics[width = 1\linewidth]{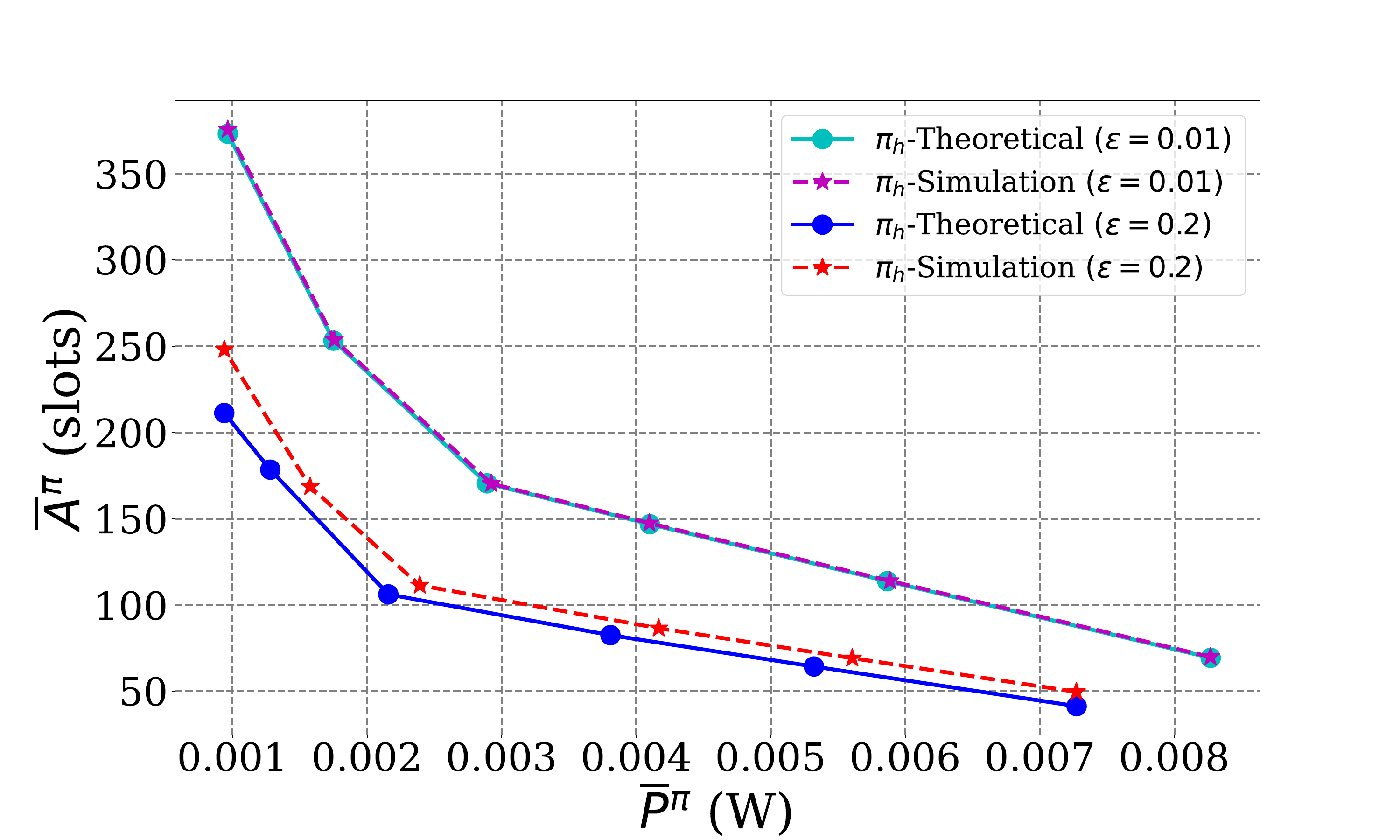}
    \caption{Tradeoff of $\overline{A}^{\pi_{h}}$ and $\overline{P}^{\pi_{h}}$}
    \label{fig:thre_tradeoff_theo_vs_sim_NP}
    \end{subfigure}
    \begin{subfigure}{0.49\linewidth}
    \includegraphics[width = 1\linewidth]{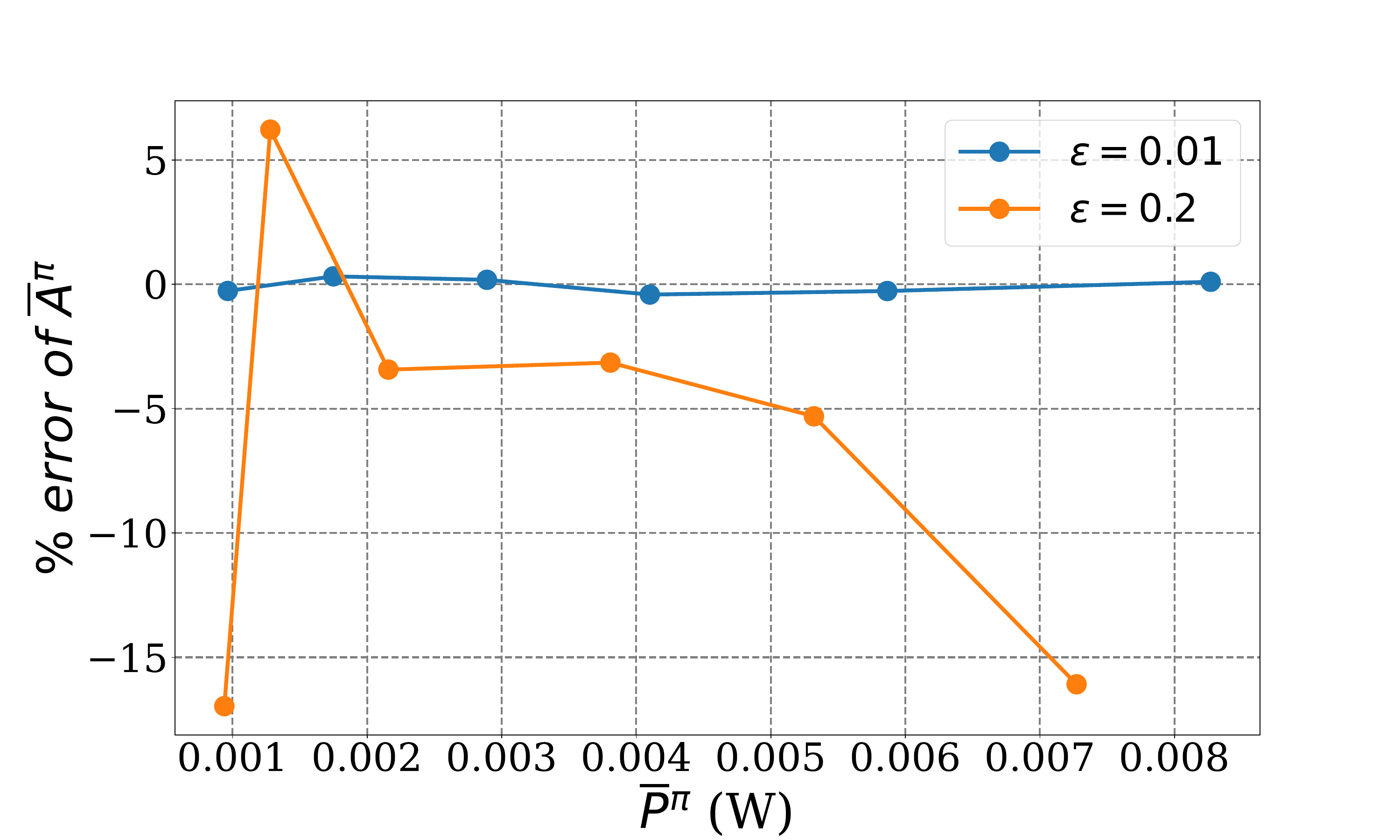}
    \caption{Percentage error in $\overline{A}^{\pi_{t_{s}}}$ calculated with respect to the simulated value}
    \label{fig:thre_tradeoff_theo_vs_sim_NP_error}
    \end{subfigure}
    \caption{Comparison of AAoI average-power tradeoff for threshold policy obtained from analysis (ERRFREE system) 
and simulation (WITHERR system) for error probabilities $\epsilon \in \brac{0.01,0.2}$. The percentage error between the analytical value and the simulated value of the AAoI calculated with respect to the simulated value at different power values is also shown.}
    \label{fig:comparison_Threshold_sim_analytical_error}
\end{figure}

We compare the analytical characterization for FTT policy in Figure \ref{fig:comparison_FTT_sim_analytical_error} for 
two error probabilities, $\epsilon \in \brac{0.01, 0.2}$.
The other parameters are chosen to be the same as in Section \ref{low_reliability_results}.
We note that the error-free analytical approximation for the average AoI at a fixed average power value has a 
percentage error which is less than $2\%$ for $\epsilon = 0.01$ while the error grows to $-20\%$ for $\epsilon = 0.2$.
The error-free analytical approximation for the tradeoff is observed to be a lower bound on the tradeoff for the system 
with errors.
We note that for the FTT policy the difference in $\overline{A}^{\pi_{t_{s}}}$ between the error-free system and 
the system with errors grows as $\frac{1}{1 - \epsilon}$ with a coefficient proportional to $\frac{1 - 
\lambda}{\lambda} 
+ t_{s}$. This can be obtained by comparing the expression for age in Proposition \ref{proposition_packet_loss} with 
that for $\epsilon = 0$.

Similarly, we compare the analytical characterization for threshold policy\footnote{We note that here there is no optimization carried out over the parameters of the threshold policy, unlike in Section \ref{low_reliability_results}.} in Figure 
\ref{fig:comparison_Threshold_sim_analytical_error} for 
two error probabilities, $\epsilon \in \brac{0.01, 0.2}$.
Again for lower error probabilities, the analytical characterization is close to the simulated values.
Thus, the analytical characterization of the average AoI and average power for the threshold policy as a function of the 
threshold and the parameters $\tau_{a}$ and $\tau_{b}$ is useful for choosing their values for the case of highly-reliable systems (i.e. small $\epsilon$).
The usage is demonstrated in Figure \ref{fig:comparison_policies}; also see the discussion.

From the tradeoff plots in Figures \ref{fig:comparison_FTT_sim_analytical_error} and \ref{fig:comparison_Threshold_sim_analytical_error} is that the AAoI for a given power constraint is more for a system which is more reliable (e.g. $\epsilon = 0.01$) compared with a less reliable (e.g. $\epsilon = 0.2$) system. This is due to larger codeword lengths being used to achieve higher reliability for a given transmit power value.
Also, the tradeoff for the ERRFREE system is a lower bound to the tradeoff for the WITHERR system.

\section{Proof of Proposition \ref{prop:Pmodel_tradeoff}}
\label{appendix:proofPmodel_tradeoff}.

We obtain the AAoI and average power for FTT policy with the P model using renewal reward theorem (RRT) \cite{kumar2008wireless}.
For applying RRT, we first identify a renewal process in the evolution of $A[t]$ under an FTT policy with parameter $t_{s}$.
We define a renewal epoch as the slot in which the age $A[t]$ drops due to a packet's successful reception.
We note that since FTT uses a fixed service time $t_{s}$, the age at a renewal epoch is $t_{s}$ since the successfully received packet was generated $t_{s}$ slots before.
We note that in the P packet-generation model a successful packet reception occurs only if there is no preemption of the transmission by a new packet and the packet reception was not in error.

We discuss about the duration of the renewal cycle in the following.
Suppose a packet is successfully received at the beginning of a slot $t_{0}$ (i.e., the last slot of the packet's transmission was from $t_{0} - 1$ to $t_{0}$).
Then, the next packet (say $p_{1}$) is generated after a random $\tilde{G}_{0}$ slots at $t_{0} + \tilde{G}_{0}$.
Here $\tilde{G}_{0} \in \brac{0,1,2,\dots}$ is a Geometric random variable with parameter $\lambda$ since we have assumed that the age has dropped at the beginning of slot $t_{0}$ and the intergeneration times are memoryless.
Packet $p_{1}$'s transmission starts at $t_{0} + \tilde{G}_{0}$ with a transmission duration of $t_{s}$.
Under P model, another packet (say $p_{2}$) is generated at a random time $t_{0} + \tilde{G}_{0} + G_{1}$ where $G_{1} \in \brac{1,2,\dots}$ is a Geometric$(\lambda)$ random variable.
We note that if $t_{0} + \tilde{G}_{0} + G_{1} < t_{0} + \tilde{G}_{0} + t_{s}$ or equivalently if $G_{1} < t_{s}$ then $p_{1}$'s transmission is stopped and $p_{2}$'s transmission starts at $t_{0} + \tilde{G}_{0} + G_{1}$.
Packet $p_{2}$'s transmission is scheduled during the slots $t_{0} + \tilde{G}_{0} + G_{1}$ to $t_{0} + \tilde{G}_{0} + G_{1} + t_{s}$.
Again, another packet (say $p_{3}$) would be generated in the slot $t_{0} + \tilde{G}_{0} + G_{1} + G_{2}$ (where $G_{i}, i \in \brac{1,2,\dots}$ are IID).
Packet $p_{2}$'s transmission is stopped and $p_{3}$'s transmission starts if $G_{2} < t_{s}$.
Thus, packet transmissions get stopped by new packets until an $i$ such that $G_{i} \geq t_{s}$.

Now, suppose we have that for an $i$, $G_{i} \geq t_{s}$.
With probability $\epsilon$, the packet transmission is in error and the age does not drop.
Then, the cycle extends for a Geometric time till the next packet generation and the process described above continues.
With probability $1 - \epsilon$ the packet is successfully received at $t_{s}$ slots after the last packet generation.

We note that if $X \in \brac{1,2,\dots}$ denotes the number of packet transmissions until a successful packet 
reception, then $X$ is a Geometric random variable with success probability being the probability of $G_{i} \geq t_{s}$ 
and successful packet reception.
If we denote this success probability by $\alpha$, then
\begin{eqnarray}
 \alpha = (1 - \epsilon)\sum_{g = t_{s}}^{\infty} \Pr{G_{i} = g} = (1 - \epsilon)(1 - \lambda)^{t_{s} - 1}.
\end{eqnarray}
Suppose we denote the renewal cycle duration by $R$. Then, we have that
\begin{eqnarray}
 R & = & \tilde{G}_{0} + \sum_{i = 1}^{X - 1} \overline{G}_{i} + t_{s}.
\end{eqnarray}
Here note that $\overline{G}_{i}$ is $G_{i}$ conditioned on the event $E_{i} = \brac{G_{i} < t_{s}}$ or $\brac{G_{i} 
\geq t_{s} \text{ and packet was received in error}}$.
The distribution of $\overline{G}_{i}$ is therefore
\begin{eqnarray*}
 \Pr{\overline{G}_{i} = g} & = & \Pr{G_{i} = g | E_{i}}, \\
 & = & \Pr{G_{i} = g \text{ and } E_{i}} / \Pr{E_{i}}, \\
 & = & \Pr{G_{i} = g \text{ and } E_{i}} / (1 - \alpha).
\end{eqnarray*}
Then, we have that
\begin{eqnarray*}
 \Pr{\overline{G}_{i} = g} & = & \frac{(1 - \lambda)^{g - 1} \lambda}{1 - \alpha}, \text{ for } g < t_{s}, \\
 & & \frac{\epsilon(1 - \lambda)^{g - 1} \lambda}{1 - \alpha}, \text{ for } g \geq t_{s}.
\end{eqnarray*}

To apply RRT, we first obtain the cumulative age $\overline{\mathcal{A}}$ in the renewal cycle.
Since the minimum age is $t_{s}$ and the age increases linearly from $t_{s}$ till the end of the renewal cycle, we 
obtain that $\overline{\mathcal{A}} = t_{s}R + \frac{R(R-1)}{2}$.
Then, using RRT, 
\begin{eqnarray*}
\overline{A}^{\pi_{t_{s}}} =\frac{\mathbb{E}[\overline{\mathcal{A}}]}{\mathbb{E}[R]} 
=t_{s}+\frac{\mathbb{E}\left[R^{2}\right]}{2 \mathbb{E}[R]}-\frac{1}{2}.
\end{eqnarray*}
We note that 
\begin{eqnarray}
\mathbb{E}[R]  =  \mathbb{E}\left[G_{0}\right]+\mathbb{E}\left[\sum_{i=1}^{X-1} \overline{G}_{i}\right]+t_{s}.
\label{eq:fttp_ER}
\end{eqnarray}
Since $G_{0} \sim \text{Geometric}(\lambda)$,  $\mathbb{E}\left[G_{0}\right]=\frac{1 - \lambda}{\lambda}$.
Also, $\Exp \bras{\sum_{i = 1}^{X - 1} \overline{G_{i}}} = \Exp \bras{\Exp \bras{\sum_{i = 1}^{x - 1} 
\overline{G}_{i}} | X }$.
We note that $\overline{G}_{i}$ is already conditioned on $E_{i}$ and is independent of $X$.
From the distribution of $\overline{G}_{i}$ we have that
\begin{eqnarray*}
 \Exp \overline{G}_{i} & = & \sum_{g = 1}^{t_{s} - 1} \frac{g(1 - \lambda)^{g - 1} \lambda}{1 - \alpha} + \epsilon 
\sum_{g = t_{s}}^{\infty} \frac{g(1 - \lambda)^{g - 1} \lambda}{1 - \alpha}, \\
& = & \sum_{g = 1}^{\infty} \frac{g(1 - \lambda)^{g - 1} \lambda}{1 - \alpha} - (1 - \epsilon) 
\sum_{g = t_{s}}^{\infty} \frac{g(1 - \lambda)^{g - 1} \lambda}{1 - \alpha}.
\end{eqnarray*}
This can be simplified to
\begin{eqnarray}
 \frac{1}{(1 - \alpha)\lambda} - \frac{1 - \epsilon}{1 - \alpha} (1 - \lambda)^{t_{s} - 1} \brap{t_{s} + \frac{1 - 
\lambda}{\lambda}}.
\label{eq:fttp_EGi}
\end{eqnarray}
Using the above expression we have that $\Exp \bras{\sum_{i = 1}^{X - 1} \overline{G_{i}}} = \brap{\frac{1}{\alpha} - 
1} \Exp \overline{G}_{i}$.
Substituting in \eqref{eq:fttp_ER} and simplifying we obtain that
\begin{eqnarray}
 \Exp R = \frac{1}{\alpha\lambda}.
\end{eqnarray}

Now, we compute $\Exp R^{2}$ as $\operatorname{Var}(R) + (\mathbb{E}[R])^{2}$.
We have that
\begin{align*}
\operatorname{Var}(R) = \operatorname{Var}\left(\tilde{G}_{0}\right)+ \operatorname{Var}\left(\sum_{i=1}^{X-1} 
\overline{G}_{i}\right),
\end{align*}
where $\operatorname{Var}\left(\tilde{G}_{0}\right) = \frac{1-\lambda}{\lambda^{2}}$. We let $S_{G}=\sum_{i=1}^{X-1} 
\overline{G}_{i}$.
Then
\begin{equation}
\operatorname{Var}\left(S_{G}\right)  =  \mathbb{E}\left[\operatorname{Var}\left(S_{G} \mid 
X\right)\right]+\operatorname{Var}\left(\mathbb{E}\left[S_{G} \mid X\right]\right).
\label{eq:fttp_SG}
\end{equation}
We consider the first term in \eqref{eq:fttp_SG}. 
\begin{align*}
\operatorname{Var}\left(S_{G} \mid X = x\right) =\operatorname{Var}\left(\sum_{i=1}^{x-1} \overline{G}_{i}\right)=(x-1) 
\cdot \operatorname{Var}\left(\overline{G}_{i}\right),
\end{align*}
so that 
\begin{align*}
\mathbb{E}[\operatorname{Var}(S_{G} \mid X)] &= \mathbb{E}\left[(X-1) 
\operatorname{Var}\left(\overline{G}_{i}\right)\right] \\
& = \operatorname{Var}\left(\overline{G}_{i}\right) \frac{1 - \alpha}{\alpha}.
\end{align*}
We note that $\operatorname{Var}\left(\overline{G}_{i}\right) = \Exp \overline{G}_{i}^{2} - \brap{\Exp 
\overline{G}_{i}}^{2}$.
Since $\Exp \overline{G}_{i}$ has been characterized in \eqref{eq:fttp_EGi} we characterize $\Exp 
\overline{G}_{i}^{2}$ in the following.
We note that
\begin{eqnarray*}
\Exp \overline{G}_{i}^{2} & = & \sum_{g = 1}^{t_{s} - 1} \frac{g^{2} (1 - \lambda)^{g - 1} \lambda}{1 - \alpha} + 
\epsilon \sum_{g = t_{s}}^{\infty} \frac{g^{2} (1 - \lambda)^{g - 1}\lambda}{1 - \alpha}, \\
& = & \sum_{g = 1}^{\infty} \frac{g^{2} (1 - \lambda)^{g - 1} \lambda}{1 - \alpha} - (1 - \epsilon) \sum_{g = 
t_{s}}^{\infty} \frac{g^{2} (1 - \lambda)^{g - 1}\lambda}{1 - \alpha}. \\
\end{eqnarray*}
We note that the first term in the above expression is the second moment of a Geometric distribution with parameter 
$\lambda$ and hence
\begin{eqnarray*}
 \sum_{g = 1}^{\infty} \frac{g^{2} (1 - \lambda)^{g - 1} \lambda}{1 - \alpha} = \frac{1}{1 - \alpha}\brap{\frac{1 - 
\lambda}{\lambda^{2}} + \frac{1}{\lambda^{2}}}.
\end{eqnarray*}
The second term $(1 - \epsilon) \sum_{g = t_{s}}^{\infty} \frac{g^{2} (1 - \lambda)^{g - 1}\lambda}{1 - \alpha}$ can be 
simplified as follows:
\begin{eqnarray*}
 & & (1 - \epsilon) \sum_{i = 0}^{\infty} \frac{(t_{s} + i)^{2} (1 - \lambda)^{t_{s} + i - 1}\lambda}{1 - \alpha}, \\
 & = & (1 - \epsilon)(1 - \lambda)^{t_{s} - 1} \sum_{i = 0}^{\infty} \frac{(t_{s}^{2} + i^{2} + 2t_{s}i)(1 - 
\lambda)^{i}\lambda}{1 - \alpha},
\end{eqnarray*}
where we note that $(1 - \epsilon)(1 - \lambda)^{t_{s} - 1} = \alpha$. Considering each of the terms (i.e. $t_{s}^{2}, 
i^{2}$ and $2t_{s}i$) in the above sum we simplify the above expression to
\begin{eqnarray*}
\frac{\alpha}{1 - \alpha} \brap{t_{s}^{2} + 2t_{s}\frac{1 - \lambda}{\lambda} + \frac{(1 - \lambda)(2 - 
\lambda)}{\lambda^{2}}}.
\end{eqnarray*}
Therefore,
\begin{eqnarray*}
 \Exp \overline{G}_{i}^{2} = \frac{1}{1 - \alpha}\brap{\frac{1 - 
\lambda}{\lambda^{2}} + \frac{1}{\lambda^{2}}} - \frac{\alpha}{1 - \alpha} \brap{t_{s}^{2} + 2t_{s}\frac{1 - 
\lambda}{\lambda} + \frac{(1 - \lambda)(2 - 
\lambda)}{\lambda^{2}}}.
\end{eqnarray*}
Thus, we obtain an analytical characterization of $\mathbb{E}[\operatorname{Var}(S_{G} \mid X)]$.

We now consider the second term in \eqref{eq:fttp_SG}.
\begin{align*}
\mathbb{E}\left[S_{G} \mid \mathrm{X}=x\right]  =  \mathbb{E}\left[\sum_{i=1}^{x-1} G_{i}\right]=(x-1) \cdot \mathbb{E}\left[G_{i}\right].
\end{align*}
Therefore,
\begin{align*}
\operatorname{Var}\left(\mathbb{E}\left[S_{G} \mid X\right]\right) =\operatorname{Var}\left((X-1) \mathbb{E}\left[G_{i}\right]\right)  =  \left(\mathbb{E}\left[G_{i}\right]\right)^{2} \operatorname{Var}(X),
\end{align*}
where $\operatorname{Var}(X) = \frac{1-\alpha}{\alpha^{2}}$.

Combining the expressions for the first and second term in \eqref{eq:fttp_SG} we have that
\begin{eqnarray*}
 \operatorname{Var}(S_{G}) = \frac{1 - \alpha}{\alpha^{2}} \brap{\Exp \overline{G}_{i}}^{2} + \frac{1 - \alpha}{\alpha} 
\operatorname{Var}(\overline{G}_{i}).
\end{eqnarray*}
Substituting for $\Exp \overline{G}_{i}$ and $\Exp \overline{G}_{i}^{2}$ and with some algebra we obtain that
\begin{eqnarray*}
 \operatorname{Var}(S_{G}) = \frac{1}{\alpha^{2}\lambda^{2}} + \frac{1 - 2 t_{s}}{\alpha\lambda} - \frac{1 - 
\lambda}{\lambda^{2}}.
\end{eqnarray*}
Now substituting in the expression for $\Exp R^{2}$ we obtain that
\begin{eqnarray}
 \Exp R^{2} = \frac{2}{\alpha^{2}\lambda^{2}} + \frac{1 - 2t_{s}}{\alpha\lambda}.
\end{eqnarray}
Finally, substituting $\Exp R$ and $\Exp R^{2}$ in the expression for $\overline{A}^{\pi_{t_{s}}}$ we have that
\begin{eqnarray}
 \overline{A}^{\pi_{t_{s}}} = \frac{1}{\alpha\lambda}.
\end{eqnarray}

Using a similar procedure, we obtain the average power using RRT.
We first obtain the total energy consumed in a renewal cycle $\overline{\mathcal{P}}$.
We have that
\begin{eqnarray*}
 \overline{\mathcal{P}} = P(t_{s})\brap{t_{s} + \sum_{i = 1}^{X - 1} \min(t_{s}, \overline{G}_{i})}.
\end{eqnarray*}
We note that there are $X$ transmissions with the last transmission having a duration of $t_{s}$.
Every other transmission is either pre-empted or lasts for a maximum of $t_{s}$, so that the transmission duration is 
$\min(t_{s}, \overline{G}_{i})$.

In order to obtain $\Exp \overline{\mathcal{P}}$ we note that $\Exp\bras{\sum_{i = 1}^{X - 1} \min(t_{s}, 
\overline{G}_{i})} = \Exp\bras{\Exp\bras{\sum_{i = 1}^{x - 1} \min(t_{s}, \overline{G}_{i})}|X}$ which is
\begin{eqnarray*}
 \brap{\frac{1 - \alpha}{\alpha}} \Exp \min(t_{s}, \overline{G}_{i}).
\end{eqnarray*}
We have from the distribution of $\overline{G}_{i}$ that
\begin{eqnarray*}
 \Exp \min(t_{s}, \overline{G}_{i}) = \frac{1}{1 - \alpha}\brap{\sum_{g = 1}^{t_{s} - 1} g(1 - \lambda)^{g - 1} \lambda 
+ \epsilon \sum_{g = t_{s}}^{\infty} t_{s} g(1 - \lambda)^{g - 1} \lambda}.
\end{eqnarray*}
This expression can be simplified as
\begin{eqnarray*}
 \frac{1}{1 - \alpha}\brap{\frac{1}{\lambda} - (1 - \lambda)^{t_{s} - 1} \brap{t_{s}(1 - \epsilon) + \frac{1 - 
\lambda}{\lambda}}}.
\end{eqnarray*}
We then have that
\begin{eqnarray*}
 \Exp\overline{\mathcal{P}} = P(t_{s})\brap{\frac{1}{\alpha\lambda} - \frac{1 - \lambda}{(1 - \epsilon)\lambda}}.
\end{eqnarray*}
The average power $\overline{P}^{\pi_{t_{s}}}$ is $\frac{\Exp \overline{\mathcal{P}}}{\Exp R}$ which can be simplified 
as
\begin{eqnarray*}
 \overline{P}^{\pi_{t_{s}}} = P(t_{s})\brap{1 - (1 - \lambda)^{t_{s}}}.
\end{eqnarray*}

\section{Proof of Proposition \ref{prop:ATtradeoff}}
\label{appendix:proofATtradeoff}
We consider the case where the age threshold $h_{a} \geq t_{s}$.
We note that with any initial value for the age $A[0]$ there exists a finite time $t$ at which $A[t] \geq h_{a}$ so that a packet is generated at $t$.
Since $\epsilon < 1$, there exists a finite time with positive probability at which a packet transmission succeeds and the age drops to $t_{s}$.
Once the age drops to $t_{s}$, we can identify a renewal process in the evolution of $A[t]$.
The renewal epochs are those slots in which the age drops to $t_{s}$.
The renewal cycles are IID with each renewal cycle being of duration
\begin{eqnarray*}
 h_{a} - t_{s} + t_{s} R,
\end{eqnarray*}
where $R \in \brac{1,2,\dots}$ is a Geometric random variable with success probability $1 - \epsilon$.
We note that $h_{a} - t_{s}$ is the time taken from the start of a renewal cycle to the time of a packet generation.
After a packet is generated and transmitted using a duration of $t_{s}$, there will be Geometric $R$ number of transmissions until a successful transmission.
The $\overline{A}^{\pi_{t_{s}}}$ and $\overline{P}^{\pi_{t_{s}}}$ are characterized using RRT.
The expected cumulative age over one renewal cycle is
\begin{eqnarray*}
 \Exp\bras{t_{s}(h_{a} - t_{s} + R)) + \frac{(h_{a} - t_{s} + R)(h_{a} - t_{s} + R - 1)}{2}},
\end{eqnarray*}
while the expected cumulative energy is
\begin{eqnarray*}
 \Exp \bras{P(t_{s})t_{s} R}.
\end{eqnarray*}
Using $\Exp R = 1/(1 - \epsilon)$ and $\Exp R^{2} = (1 + \epsilon)/(1 - \epsilon)^{2}$ and applying RRT we obtain that
\begin{eqnarray*}
 \overline{A}^{\pi_{t_{s}}} & = & t_{s} + \frac{(h_{a}-t_{s})^{2}(1 - \epsilon)^{2} + t_{s}^{2}(1 + \epsilon) + 2(h_{a} - t_{s})t_{s}(1 - \epsilon)}{2(1 - \epsilon)((h_{a}-t_{s})(1 - \epsilon) + t_{s})} - \frac{1}{2}, \\
 \overline{P}^{\pi_{t_{s}}} & = & \frac{P(t_{s})t_{s}}{(h_{a} - t_{s})(1 - \epsilon) + t_{s}}
\end{eqnarray*}

We note that if $h_{a} < t_{s}$, then after each successful packet transmission another packet is immediately generated so that the renewal cycle length is just $t_{s}R$.
Thus, $\overline{A}^{\pi_{t_{s}}}$ and $\overline{P}^{\pi_{t_{s}}}$ can be obtained using $h_{a} = t_{s}$ in the above expressions.

\bibliographystyle{ieeetr}
{\footnotesize
\bibliography{IEEEabv, bibJournalList, bibliography.bib}}
\end{document}